\documentclass[num-refs]{wiley-article}
\usepackage[superscript, biblabel]{cite}
\usepackage[linesnumbered, lined, boxed]{algorithm2e}
\usepackage{siunitx}
\usepackage{setspace} 
\usepackage{lineno}\modulolinenumbers[5]
\usepackage{algpseudocode}
\usepackage{graphicx}
\usepackage[markers,figuresonly,nolists]{endfloat}
\usepackage{caption, subcaption}

\let\emptyset\varnothing
\usepackage{tabularx, multicol, nomencl}\makenomenclature

\definecolor{ncsuaqua}{RGB}{0, 132, 115} 
\definecolor{ncsublue}{RGB}{65, 86, 161}
\definecolor{ncsured}{RGB}{206, 0, 0}
 
\usepackage[colorlinks=true, linkcolor=ncsublue, urlcolor=ncsublue]{hyperref}
\usepackage{url}

\usepackage{amsmath}
\usepackage{algorithm}
\usepackage{algpseudocode}

\usepackage{IEEEtrantools}

\long\def\comment#1{}

\newcommand{\be}{\begin{equation}}
\newcommand{\ee}{\end{equation}}
 \newtheorem{prop}{Proposition}

\newcommand{\eqdef}{\stackrel{\Delta}{=}}

\newcommand{\av}{{\bf a}}

\newcommand{\dv}{{\bf d}}
\newcommand{\ev}{{\bf e}}

\newcommand{\nv}{{\bf n}}

\newcommand{\uv}{{\bf u}}

\newcommand{\vvv}{{\bf v}}
\newcommand{\xv}{{\bf x}}
\newcommand{\yv}{{\bf y}}

\newcommand{\Upsilonm}{{\bf \Upsilon}}

\newcommand{\Am}{{\bf A}}
\newcommand{\Bm}{{\bf B}}
\newcommand{\Cm}{{\bf C}}

\newcommand{\Km}{{\bf K}}
\newcommand{\Lm}{{\bf L}}
\newcommand{\Mm}{{\bf M}}

\newcommand{\Pm}{{\bf P}}
\newcommand{\Qm}{{\bf Q}}

\newcommand{\Sm}{{\bf S}}

\newcommand{\Wm}{{\bf W}}

\newcommand{\Ac}{{\cal A}}

\newcommand{\Hc}{{\cal H}}

\newcommand{\Kc}{{\cal K}}

\newcommand{\Nc}{{\cal N}}

\newcommand{\Sc}{{\cal S}}

\newcommand{\RNum}[1]{\uppercase\expandafter{\romannumeral #1\relax}}

\newcommand{\epsilonv}{\hbox{\boldmath$\epsilon$}}

\newcommand{\muv}{\hbox{\boldmath$\mu$}}

\newcommand{\xiv}{\hbox{\boldmath$\xi$}}

\newcommand{\Deltam}{\hbox{\boldmath$\Delta$}}
\newcommand{\Sigmam}{\hbox{\boldmath$\Sigma$}}

\DeclareMathOperator*{\argmin}{arg\,min}

\papertype{Process Systems Engineering}
\title{Process Resilience under Optimal Data Injection Attacks}
\author[1]{Xiuzhen Ye}
\author[1]{Wentao Tang}
\affil[1]{Department of Chemical and Biomolecular Engineering, North Carolina State University}
\corraddress{Wentao Tang, Campus Box 7905, North Carolina State University, Raleigh, NC 27695 U.S.A.}
\corremail{\href{mailto:wtang23@ncsu.edu}{wentao\_tang@ncsu.edu}}
\fundinginfo{This work is supported by Wentao Tang's startup fund from NC State University.}
\runningauthor{Ye \& Tang}

\begin{document}
\maketitle
\begin{abstract}
{In this paper, we study the resilience of process systems in an {\it information-theoretic framework}, from the perspective of an attacker capable of optimally constructing data injection attacks.} The attack aims to distract the stationary distributions of process variables and stay stealthy, simultaneously. The problem is formulated as designing a multivariate Gaussian distribution to maximize the Kullback-Leibler divergence between the stationary distributions of states and state estimates under attacks and without attacks, while minimizing that between the distributions of sensor measurements. 
When the attacker has limited access to sensors, sparse attacks are proposed by incorporating a sparsity constraint. {We conduct theoretical analysis on the convexity of the attack construction problem and present a greedy algorithm, which enables systematic assessment of measurement vulnerability, thereby offering insights into the inherent resilience of process systems.
We numerically evaluate the performance of proposed constructions on a two-reactor process. }
\end{abstract}

\section{Introduction} 
In modern chemical processes, computational and physical components are tightly integrated, which turns traditional chemical processes into {\it{cyber-physical systems} }(CPSs)~\cite{Alur_PrinciplesCPS_2015}. The integration paves the way for new technologies from artificial intelligence, and more broadly, data-driven methods. A representative example of such tightly integrated systems is the supervisory control and data acquisition (SCADA) system that monitors, controls and manages critical infrastructure~\cite{HR_IoT_17}. However, these advances and potentials are harnessed at the expense of emerging challenges~\cite{RH_ICT_11}. {Therefore, as new techniques are developed for tightly integrated systems, ensuring the resilience of the processes has become a critical concern. Investigations into process resilience require a fundamental understanding of the threats these systems are exposed to, especially {\it cyberattacks} due to their impact on data acquisition and transfer across control networks, communication channels, and embedded computations.} Such malicious attacks disrupt, manipulate, and exploit the interactions between physical processes and computational control systems. As chemical processes are highly susceptible to cyberattacks from various entry points~\cite{YG_IoT_2023}, cyberattacks in chemical processes, such as data injection attacks (DIAs), replay attacks, denial of service (DoS), Man-in-the-Middle (MitM), and others, have been investigated in both research and public sectors~\cite{YM_CCC_2009, DZ_JAS_2022}. In the context of chemical processes as nonlinear systems~\cite{Helen_Mathematics_2018, KeHelen_CERD_2021}, typically controlled with multivariate MPC~\cite{WuHelen_Mathematics_2018}, issues such as attack detection~\cite{OyamaHelen_AICHE_2020} and the impact of cyberattacks on specific applications have been studied in a few recent works~\cite{HelenMatthew_Mathematics_2020, Huang_TIT_2018}. 
 
{A major security concern in process resilience that needs to be addressed is DIAs~\cite{LY_TISSEC_11}, which, however, have received limited attention in the literature.}
DIAs are such attacks that affect feedback control and alter the system behavior by compromising the measurements in acquisition and communication without triggering detection mechanisms~\cite{DH_TSMC_2020}. More directly, DIAs can also introduce disturbances into the states of the process, which drastically impacts the efficiency, safety, and product quality. This kind of attacks does not entail physical damages to the system, which is much less costly for the attacker. Consequently, it has become one of the most major cyberattacks to CPSs~\cite{WS_TC_2018}.
Chemical processes are sensitive to DIAs, as the dynamics rely heavily on data availability for control systems, where feedback control plays a critical role in exploiting the real-time data and optimizing the system behavior. Data-driven control methods, which have received increasing attention in recent literature~\cite{WT_ACC_2022}, further intensifies the reliance on data integrity and urges systematic studies on DIAs on such systems.

The research on cyberattacks in control frameworks mainly focuses on the analysis, detection, and resilience to failures. {A comprehensive review of cybersecurity challenges and potential solutions in process control, operations, and supply chain in process engineering can be found in Parker~\cite{PW_CCE_23}, where cyber-attack detection and resilient control are discussed.}
The first work in the scope of control theory for cyberattacks was by Mo \& Sinopoli~\cite{YM_SCS_2010}, where the authors analyzed the effects of such attacks on control systems and provided a necessary and sufficient condition under which the attacks could destabilize the system, while successfully bypassing a large set of possible failure detectors. 
In Pasqualetti et al.~\cite{FabioP_CDC_2011}, notions of detectability and identifiability of an attack in control systems are introduced, as well as the corresponding dynamical detection and identification procedures based on tools from geometric control theory. In an extended work of Pasqualetti et al.~\cite{FabioP_TAC_2013}, the centralized and distributed attack detection and identification monitors are proposed. In these works, the dynamics of the system is described as a linear state-space model. Based on the same formula, DIAs to the load of a controlled energy system against stability are formulated and analyzed~\cite{FabioP_TSG_2018}.  
Furthermore, controllability and observability in control theory can be used for the investigation of security in the systems, and they ultimately lead to security-aware system design criteria. 
For chemical processes, the work by Durand~\cite{Helen_Mathematics_2018} developed a framework to guide the process design to be cyberattack resilient. In the work by  Rangan et al.~\cite{KeHelen_CERD_2021}, cyberattack detection strategies for nonlinear systems under particular control strategy are discussed with a designed detection criteria.
The scope of the framework is restricted to specific control laws, MPC, and residual-based detection methods~\cite{ShilpaEllis_ACC_2022, ShilpaEllis_JPC_2022}. 
{Seeking a more generic approach, neural network based detection methods were proposed by Wu et al.~\cite{WuHelen_Mathematics_2018}~\cite{Huang_CSR_2020}, which appears to, however, require a significant amount of labeled training data and lack generalizability to novel attacks. A recent study by Wu et al.\cite{WZ_IECR_25} addressed this limitation by incorporating information from process networks and prior knowledge of the patterns/characteristics of specific cyberattacks. Similarly, the work by Wu et al.~\cite{WW_CERD_24} proposed physics-informed machine learning detection methods, which are subject to real-time deployment.}

Apart from DIAs targeting at sensors, cyberattacks on actuators manipulate the control inputs transmitted across control channels, which yield {\it{input attacks}}.  
In the work of Bai et al.~\cite{FabioP_ACC_2015}, the proposed $\epsilon$-stealthy attacks tamper with the control inputs in stochastic control systems to maximize the estimation error of the Kalman filter, which gives an achievable bound and a closed-form expression of the attacks.
The concept of $\epsilon$-stealthiness means that no detector, with a certain bounded probability of detection, can obtain a probability of false alarm that exponentially converges to zero at a rate greater than $\epsilon$.  
The work was extended later to quantify the estimation error and characterize the limitations of the stealthiness of an input attack~\cite{FabioP_Automatica_2017}.
The attack strategies rely heavily on the assumption of right-invertible systems, and for systems that are not right-invertible, only suboptimal strategies and looser bounds are provided.

We remark that the above-mentioned works are on stochastic systems.
In this setting, {\it information theory} offers quantitative measures of the information in the data~\cite{TM_ElementsofIT}, and thus excels in establishing an universal framework to fundamentally characterize the effects of DIAs~\cite{YE_TSG_21}.
Studies on information-theoretic DIAs were presented in the previous works by the first author~\cite{YE_TSG_21, YE_SGC_20} on smart grids, where the attack detection is formulated as the likelihood ratio test~\cite{JN_LRT_33} or alternatively machine learning techniques~\cite{OM_TNNLS_16}. Apart from attack constructions, information-theoretic measures also serve as fundamental metrics in assessing sensor vulnerabilities~\cite{YE_IETSG_22}. 

{This paper studies process resilience under optimal data injection attacks (DIAs) within an information-theoretic framework. The attacks compromise sensor measurements with the aim of perturbing the distributions of system states and their estimates obtained via a state observer, while remaining stealthy.}
In addition, we aim to characterize the fundamental limits of the threats and the vulnerabilities of the sensor measurements. Here, considering the stationary distributions of process variables, the DIAs are cast in an information-theoretic framework~\cite{TM_ElementsofIT}, where the aims of the attacks are two-folded: (1) the maximum deviation of the stationary distribution under attacks from the stationary distribution without attacks, and (2) the minimum probability of attack detection.  
With such aims, the attack construction is formulated in a general setting, followed with a sparsity constraint on the numbers of sensor measurements under attacks, yielding $k$-{\it sparse attacks}. 

The main contributions of this paper follow:
\begin{enumerate}
	\item An information-theoretic framework for DIAs is proposed on control systems. The cost of the attacks is firstly formulated in terms of the information-theoretic description of the disruption and detection as a two objective optimization problem. Specifically, the Kullback-Leibler (KL) divergence between the stationary distributions with attacks and without attacks captures the attack disruption, and the KL divergence between the normal-operation sensor measurements and the attacked measurements captures the attack detection. A weighting parameter is adopted for the trade-off between the attack disruption and detection, allowing the attacker to construct attacks with customized affordable probability of detection.
	\item In a general attack setting, we first consider the case where the attacker has access to all the measurements in the system, which yields {\it full attacks}. The corresponding convexity analysis are carried out and an explicit solution is characterized. 
	\item Then, a sparsity constraint is considered where the attacker has limited access to the measurements, which yields $k$-{\it sparse attacks}. We tackle such a combinatorial problem by incorporating one-at-a-time additional measurement that yields a sequential
	sensor selection problem. In the single attack case where the attacker can only attack one measurement, the strategy is to characterize the optimal attack for all the measurements individually and to identify the one with the best performance from the attacker's standpoint. The convexity of the resulting optimization problem is analyzed, and the strategy obtained from optimal single measurement attack is distilled to propose a heuristic greedy algorithm for $k$-{\it sparse attacks}.
	\item Based on sparse attack constructions in control systems, we propose to characterize the vulnerability of process measurements by the Pareto curve of the KL divergence on states and state estimates and the KL divergence on sensor measurements. 
	The Pareto curve assesses the achievable attack disruption and attack detection for each measurement, as well as for each unit in a chemical process by allowing attacks on all the measurements in the unit. Illustrative examples in a two-reactor system are presented numerically.
\end{enumerate}
 
The rest of the paper is organized as follows. In the next section, we introduce DIAs on control systems with linearized dynamics, followed with a section where information-theoretic attacks are formulated. Full attack constructions and $k$-sparse attacks are introduced in the sections after the attack formulation section, respectively. 
We evaluate the performance of the proposed attack constructions numerically in a separate section and close the paper with conclusions in the last section.

\textbf{Notation:} 
The set of positive real numbers is denoted by $\mathbb{R}_+$. The set of $n$-dimensional vectors is denoted by $\mathbb{R}^n$. 
The elementary vector $\ev_i\in\mathbb{R}^n$ is a vector of zeros except for a one in the $i$-th entry. 
%
A real matrix $\Am$ with dimension $m \times n$ is denoted by $\Am \in \mathbb{R}^{m \times n}$.
The set of positive semidefinite matrices of size $n$ is denoted by $\Sc_{+}^n $. The cardinality of a set $\Ac$ is denoted by $|\Ac|_0$. We also denote positive semidefinite matrix $\Am$ as $\Am \succeq \textnormal{\textbf{0}}$ and positive definite matrix $\Am$ as $\Am \succ \textnormal{\textbf{0}}$. The $n$-dimensional identity matrix is denoted as $\textbf{I}_n$.
The determinant of a square matrix $\Am$ is denoted by $|\Am|$, the trace of the matrix is denoted by $\textnormal{tr}(\Am)$ and the $i$-th eigenvalue of the matrix is denoted by $\lambda_i(\Am)$.
Given a vector $\muv \in \mathbb{R}^n$ and a matrix $\Sigmam \in \Sc_{+}^n$, we denote by $\mathcal{N} (\muv, \Sigmam)$ the multivariate Gaussian distribution with mean $\muv$ and covariance matrix $\Sigmam$.
The Kullback-Leibler (KL) divergence from distribution $P$ to $Q$ is denoted by $D(P\| Q)$.
We denote the Kronecker product of $\xv$ and $\yv$ as $\xv \otimes \yv$ and the vectorization of a matrix $\Am$ as $\textnormal{vec}(\Am)$.

\section{System Model}\label{sec_system_model}
\subsection{Dynamical Model and Attack}
In this paper, we consider a discrete-time linear system with noise. A more general framework for nonlinear systems is more challenging, as we shall see, due to the involvement of stationary distributions that are usually non-Gaussian. Although chemical processes are in principle nonlinear, we restrict to linear systems as approximations near steady states in the current paper.  
The system is represented as: 
\begin{IEEEeqnarray}{ll}\label{obs_model}
\begin{aligned} 
\xv(t+1) = &\Am \xv(t) + \Bm \uv(t) + \vvv_d(t), \\
\yv(t+1) = &\Cm \xv(t)+\vvv_n(t),
\end{aligned}
\end{IEEEeqnarray}
where the vector $\xv(t) \in \mathbb{R}^n$ is the states of the control system at time $t$; the vector $\uv(t) \in \mathbb{R}^m$ is the control inputs at time $t$; the vector of output measurements at time $t$ is denoted by $\yv(t) \in \mathbb{R}^m$; the matrices $\Am \in \mathbb{R}^{n \times n}$, $\Bm \in \mathbb{R}^{n \times m}$ and $\Cm \in \mathbb{R}^{m \times n}$ are the transition matrix, input matrix and measurement matrix, respectively. In this paper, we consider the system to be observable and controllable, that is, both the observability matrix 
${\cal O} = \begin{bmatrix}
   \Cm \\
   \Cm\Am \\
   \vdots \\
   \Cm \Am^{n-1}
   \end{bmatrix}$ 
and controllability matrix $ {\cal C} = \left[ \Bm, \Am\Bm, \Am^2 \Bm, ..., \Am^{n-1} \Bm \right] $
are of rank $n$.
The output measurements $\yv$ are corrupted by additive white Gaussian noise introduced in sensing. Such a noise is modelled by the vector $\vvv_n \in \mathbb{R}^m$ in~\eqref{obs_model}, such that $\vvv_n \sim {\cal N} \left(\textnormal{\textbf{0}}, \Sigmam_{\nv \nv } \right)$.
and the disturbance to the state of the system is modelled as
$\vvv_d \sim {\cal N} \left(\textnormal{\textbf{0}}, \Sigmam_{\dv \dv} \right)$,
where $\textnormal{\textbf{0}}$ is a zero mean vector; $\Sigmam_{\nv \nv}$ and $ \Sigmam_{\dv \dv}$ are the covariance matrices of the system noise and disturbance, respectively, and it holds that $\Sigmam_{\nv \nv} \in \Sc_+^m$ and $ \Sigmam_{\dv \dv} \in \Sc_+^m$.

Consider a state observer $\hat{\xv} (t+1) = \Am \hat{\xv} (t) + \Bm \uv(t) + \Lm (\yv(t) -\Cm  \hat{\xv}(t) )$, specified by the observer gain $\Lm \in \mathbb{R}^{n \times m}$, where $\hat{\xv} (t)$ is the state estimates at time $t$ and a feedback control law $\uv(t) = \Km \hat{\xv}(t)$ is applied. From~\eqref{obs_model}, the states and the estimates satisfy
\begin{IEEEeqnarray}{ll}\label{estimator}
\begin{aligned} 
 \xv(t+1) = & \Am \xv(t) +  \Bm \Km \hat{\xv} (t) + \vvv_d(t),\\
\hat{\xv} (t+1) = & \Lm \Cm  \xv (t) + (\Am - \Lm \Cm + \Bm \Km) \hat{\xv} (t) + \Lm \vvv_n(t).
\end{aligned}
\end{IEEEeqnarray} 

{In DIAs, the attacker compromises the sensor measurement $\yv$ in~\eqref{obs_model} by a malicious attack vector $\av \in \mathbb{R}^m$ such that $\av$ has a distribution $ P_{\av}$~\cite{LY_TISSEC_11}. The injected attack vector $\av$ results in an additional term in the state observer and leads to derivations in both the dynamics of states and state estimates in~\eqref{estimator}.    }
In the following, $P_{\av}$ is assumed to be a multivariate Gaussian distribution, i.e.,
$\av \sim {\cal N}(\muv_a, \Sigmam_{\av \av})$,
where $\muv_a \in \mathbb{R}^m$ and $\Sigmam_{\av \av} \in \Sc_+^m$ are the mean vector and the covariance matrix of the attacks. Consequently, the resulting measurements denoted by $\yv_a \in \mathbb{R}^m$ are
\be\label{obs_measurements}
\yv_a = \Cm \xv +\vvv_n + \av.
\ee
The stationary distribution of the system is a multivariate Gaussian one with a zero mean vector and a covariance matrix $\Sigmam_{\xv \xv}$:
$\xv \sim {\cal N}(\textnormal{\textbf{0}}, \Sigmam_{\xv \xv})$,
where $\Sigmam_{\xv \xv} \in \Sc_+^n$. Hence, the vector of measurements $\yv$ in~\eqref{obs_model} follows a multivariate Gaussian distribution with a zero mean vector and covariance matrix $\Sigmam_{\yv\yv}$ such that
\be\label{dis_y}
\yv \sim {\cal N}(\textnormal{\textbf{0}}, \Sigmam_{\yv \yv}),\ \textnormal{with} \ \Sigmam_{\yv \yv} \eqdef \Cm \Sigmam_{\xv \xv} \Cm^{\sf T} + \Sigmam_{\nv \nv },
\ee
whereas, from~\eqref{obs_measurements}, the vector of compromised measurements is 
\be\label{dis_ya}
\yv_a \sim {\cal N}(\muv_a, \Sigmam_{\yv_a \yv_a}), \ \textnormal{with} \ \Sigmam_{\yv_a \yv_a} \eqdef \Cm \Sigmam_{\xv \xv} \Cm^{\sf T} + \Sigmam_{\nv \nv } + \Sigmam_{\av \av}.
\ee 
{ The additive DIAs on the measurements $\yv$ leads to an additive covariance matrix $\Sigmam_{\av \av}$ to the covariance matrix of the measurements without attacks, that is, $\Sigmam_{\yv_a \yv_a} = \Sigmam_{\yv \yv } + \Sigmam_{\av \av}$. 
	
Next we analyze how the attack vector $\av$ misleads the state estimator and causes the derivations in both the dynamics of states and state estimates in~\eqref{estimator}.} We denote the states under attacks by $\xv_a \in \mathbb{R}^n$ and the corresponding estimates $\hat{\xv}_a \in \mathbb{R}^n$. Let $ \xiv_a \eqdef \begin{bmatrix} 
 {\xv}_a \\ 
\hat{ {\xv}}_a  
\end{bmatrix} $ be their joint vector.
From~\eqref{obs_model},~\eqref{estimator}, and~\eqref{obs_measurements}, the dynamics of the states and the corresponding estimates under attacks are as follows:
\begin{gather}\label{20250107_1}
\xiv_a(t+1)
= \begin{bmatrix}
   \Am  & \Bm \Km\\
   \Lm \Cm & \Am - \Lm \Cm + \Bm \Km
   \end{bmatrix}
\xiv_a(t) +
\begin{bmatrix} 
\vvv_d(t) \\ 
\Lm (\vvv_n(t) +  \av(t) )
\end{bmatrix}.
\end{gather}
We denote the covariance matrix of $\xiv_a$ at the stationary distribution by $\Sigmam_{\xiv_a  \xiv_a  }\in \mathbb{R}^{2n}$. Then, it holds that
\begin{IEEEeqnarray}{ll}\label{Sigma_xiaxia}
\begin{aligned} 
\Sigmam_{\xiv_a  \xiv_a  } =\begin{bmatrix}
   \Am  &  \Bm \Km\\
   \Lm \Cm  & \Am - \Lm \Cm + \Bm \Km
   \end{bmatrix} \Sigmam_{\xiv_a \xiv_a  } \begin{bmatrix}
   \Am  &  \Bm \Km\\
   \Lm \Cm  & \Am - \Lm \Cm + \Bm \Km
   \end{bmatrix}^{\sf T} + \begin{bmatrix} 
\Sigmam_{\dv \dv} & \textnormal{\textbf{0}} \\ 
\textnormal{\textbf{0}} & \Lm ( \Sigmam_{\nv \nv }+\Sigmam_{\av \av})  \Lm^{\sf T}  
\end{bmatrix},
\end{aligned}
\end{IEEEeqnarray}
where $\textnormal{\textbf{0}} $ here denotes the zero matrix with an appropriate dimension. Let $ \xiv \eqdef \begin{bmatrix} 
 {\xv} \\ 
\hat{ {\xv}} 
\end{bmatrix} $ be the vector of the states and state estimates without attacks, and denote its stationary covariance matrix by $\Sigmam_{\xiv \xiv}$. Thus, it follows that
\begin{IEEEeqnarray}{ll}\label{Sigma_xixi} 
\begin{aligned} 
\Sigmam_{\xiv   \xiv } =\begin{bmatrix}
   \Am  &  \Bm \Km\\
   \Lm \Cm  & \Am - \Lm \Cm + \Bm \Km
   \end{bmatrix} \Sigmam_{\xiv \xiv  } \begin{bmatrix}
   \Am  &  \Bm \Km\\
   \Lm \Cm  & \Am - \Lm \Cm + \Bm \Km
   \end{bmatrix}^{\sf T} + \begin{bmatrix} 
\Sigmam_{\dv \dv} & \textnormal{\textbf{0}} \\ 
\textnormal{\textbf{0}} & \Lm \Sigmam_{\nv \nv } \Lm^{\sf T}  
\end{bmatrix}.
\end{aligned}
\end{IEEEeqnarray}

From the comparison between~\eqref{Sigma_xiaxia} and~\eqref{Sigma_xixi}, DIAs on control systems explicitly introduce the covariance matrix $\Sigmam_{\av \av}$ into the dynamics and the difference between the covariance matrices becomes 
\be
\nonumber
\Sigmam_{\xiv_a   \xiv_a } - \Sigmam_{\xiv   \xiv } = \begin{bmatrix}
   \Am  &  \Bm \Km\\
   \Lm \Cm  & \Am - \Lm \Cm + \Bm \Km
   \end{bmatrix} (\Sigmam_{\xiv_a   \xiv_a } - \Sigmam_{\xiv   \xiv }) \begin{bmatrix}
   \Am  &  \Bm \Km\\
   \Lm \Cm  & \Am - \Lm \Cm + \Bm \Km
   \end{bmatrix}^{\sf T} + \begin{bmatrix} 
\textnormal{\textbf{0}} & \textnormal{\textbf{0}} \\ 
\textnormal{\textbf{0}} & \Lm \Sigmam_{\av \av}  \Lm^{\sf T}  
\end{bmatrix}.
\ee
This causes the covariance matrices $\Sigmam_{\xv_a \xv_a}$ and $\Sigmam_{\hat{\xv}_a \hat{\xv}_a}$ of the stationary distribution under attacks to deviate from the corresponding distribution without attacks. From a stealthy attacker's point-of-view, the attack also needs to avoid being detected. Hence, the attack detection mechanism needs to be accounted for when designing the attacks.
 
\subsection{Attack Detection}
We suppose that the system operator should have security strategies in place, prior to performing state estimation. 
Attack detection is cast as a hypothesis testing problem given by a null hypothesis $\mathcal{H}_0$ and the alternative hypothesis $\mathcal{H}_1$:
\begin{subequations}\label{EqHypTestA}
\begin{align}
  \mathcal{H}_0&:\textnormal{There is no attack;}\\ 
  \mathcal{H}_1&:\textnormal{Sensor measurements are compromised}.
\end{align}
\end{subequations}
At each time $t$, the detector acquires a vector of measurements denoted by $\bar{\yv}(t)$ and decides whether this vector is obtained from the system without attacks in~\eqref{obs_model} or under attacks in~\eqref{obs_measurements}. Given the Gaussianity in the stationary distribution, the hypothesis test is done in terms of the probability density functions of the state variables, the system noise, and the attack on the measurements. Hence, the hypotheses in \eqref{EqHypTestA} become 
\begin{subequations}\label{eq:hypoth_attack}
\begin{align}
  \mathcal{H}_0&: \bar{\yv} \thicksim P_{\yv}, \\
  \mathcal{H}_1&: \bar{\yv} \thicksim P_{\yv_a}.
\end{align}
\end{subequations}
An attack detection procedure $T$ takes the measurements $\bar{\yv}$ and determines if the measurements are under attacks, that is, $T:\mathbb{R}^m\rightarrow\{0,1\}$. Let $T(\bar{\yv} )=0$ denote the case where the attack detection decides on $\Hc_0$ upon the measurements $\bar{\yv}$; and $T(\bar{\yv} )=1$ decides on $\Hc_1$.
The performance of an attack detection is assessed in terms of the Type-I error, denoted by $\alpha\eqdef P\left[T\left(\bar{\yv}\right)=1\right]$ when $\bar{\yv}\thicksim P_{\yv}$; and the Type-II error, denoted by $\beta\eqdef P\left[T\left(\bar{\yv}\right)=0 \right]$ when $\bar{\yv}\thicksim P_{\yv_a}$.
Given the requirement that the Type-I error satisfies $\alpha \leq  \alpha'$, with $\alpha'\in[0,1]$, the likelihood ratio test (LRT) is optimal in the sense that it induces the smallest Type-II error $\beta$~\cite{JN_LRT_33}. In this setting, the LRT is given by
\be\label{lrt}
T(\bar{\yv}) = 
\begin{cases} 
1, & \text{if } L(\bar{\yv}) \geq \tau, \\
0 & \text{else}.
\end{cases}
\ee
with $L(\bar{\yv})$ the likelihood ratio:
\begin{equation}\label{lr}
L(\bar{\yv}) = \frac{f_{\yv_a}(\bar{\yv})}{f_{\yv}(\bar{\yv})},
\end{equation}
where the functions $f_{\yv_a}$ and $f_{\yv}$ are the probability density function (PDF) of $\yv_a$ in~\eqref{dis_ya} and the PDF of $\yv$ in~\eqref{dis_y}, and $\tau\in\mathbb{R}_+$ in~(\ref{lrt}) is the decision threshold. 
%
{The attack detection and the deviation caused on the covariance matrices $\Sigmam_{\hat{\xv} \hat{\xv} }$ and $\Sigmam_{ \xv \xv}$ are the two objectives both for the detector and for the attacker. We emphasize that the LRT is an optimal attack detection method in the sense that among all the detection methods with the same Type-I error, LRT achieves the smallest Type-II error. Even though, in this paper, we study the DIAs construction from attacker's standpoint, an optimal attack detection method provides the best achievable probability of detection for the system operator and also evaluates the stealthiness of the attack construction. }
 
\section{Information-theoretic attacks}\label{sec_Information_Theoretic_Metrics}
\subsection{Information-Theoretic Measure}
In the previous section, as an additive attack vector to measurements $\yv$, the attack vector $\av$ leads to a deviation in the covariance matrix of stationary distribution $P_{\xiv}$. We study DIAs from the attacker's standpoint in this paper. The rationale is that the attack strategies provide the insights on the system vulnerabilities and foundations for the potential protection mechanism. The aims of the attacker are two-fold: (1) stationary distribution $P_{\xiv}$, and (2) staying stealthy. 

Consequently, to develop a universal framework for DIAs, an information-theoretic criterion is adopted. 
Specifically, to quantity the derivation caused by the attacks, {the attacker maximizes the Kullback-Leibler (KL) divergence~\cite{TM_ElementsofIT} between the stationary distributions under attacks and without attacks, denoted by $P_{\xiv_a}$ 
and $P_{\xiv}$, respectively, with $\xiv_a \eqdef (\xv_a, \ \hat{\xv}_a)$ and $\xiv \eqdef 
(\xv, \ \hat{\xv} )$.} 
The KL divergence from $P_{\xiv_a}$ to $P_{\xiv}$ is defined as~\cite{TM_ElementsofIT}:
\be
D(P_{\xiv_a} \|P_{\xiv}) \eqdef \mathbb{E}_{x \sim P_{\xiv_a}} \left[ \log \frac{P_{\xiv_a}(x)}{P_{\xiv}(x)} \right].
\ee

{On the other hand, we adopt the approach in Ye et al.~\cite{YE_TSG_21} to promote the stealth of the attacks by minimizing $D(P_{\yv_a} \|P_{\yv})$, where $P_{\yv_a}$ and $P_{\yv}$ are given in~\eqref{dis_ya} and~\eqref{dis_y}, respectively.} 
DIAs are therefore constructed as a probability distribution $P_{\av}$, which is the solution to the following optimization problem:
\begin{IEEEeqnarray}{ll}\label{obj}
\begin{aligned} 
\min_{\av \sim P_{\av}} & \ \  - D(P_{\xiv_a}   || P_{\xiv}) +\lambda D(P_{\yv_a} \| P_{\yv}),
\end{aligned}
\end{IEEEeqnarray}
where the domain of $P_{\av}$ is the set of Gaussian distributions with dimension $m$; the weighting parameter $\lambda > 0$ determines a tradeoff between the following two attack objectives: (1) the attack disruption $ D(P_{\xiv_a} \| P_{\xiv})$, and (2) attack detection $D(P_{\yv_a} \| P_{\yv})$. Note that larger values of $ D(P_{\xiv_a} \| P_{\xiv})$ indicates that the distribution of steady state and the corresponding estimates under attacks deviates more from the distribution without attacks. The attacker aims to obtain large value of $ D(P_{\xiv_a} \| P_{\xiv})$ while maintain a small $D(P_{\yv_a} \| P_{\yv})$ to reduce the probability of detection. 
{The idea that the optimal attack construction by the attacker takes into account these two terms are conceptually illustrated in Fig. \ref{fig:diag}.} 
\begin{figure}[ht]
\centering
\includegraphics[width=\textwidth]{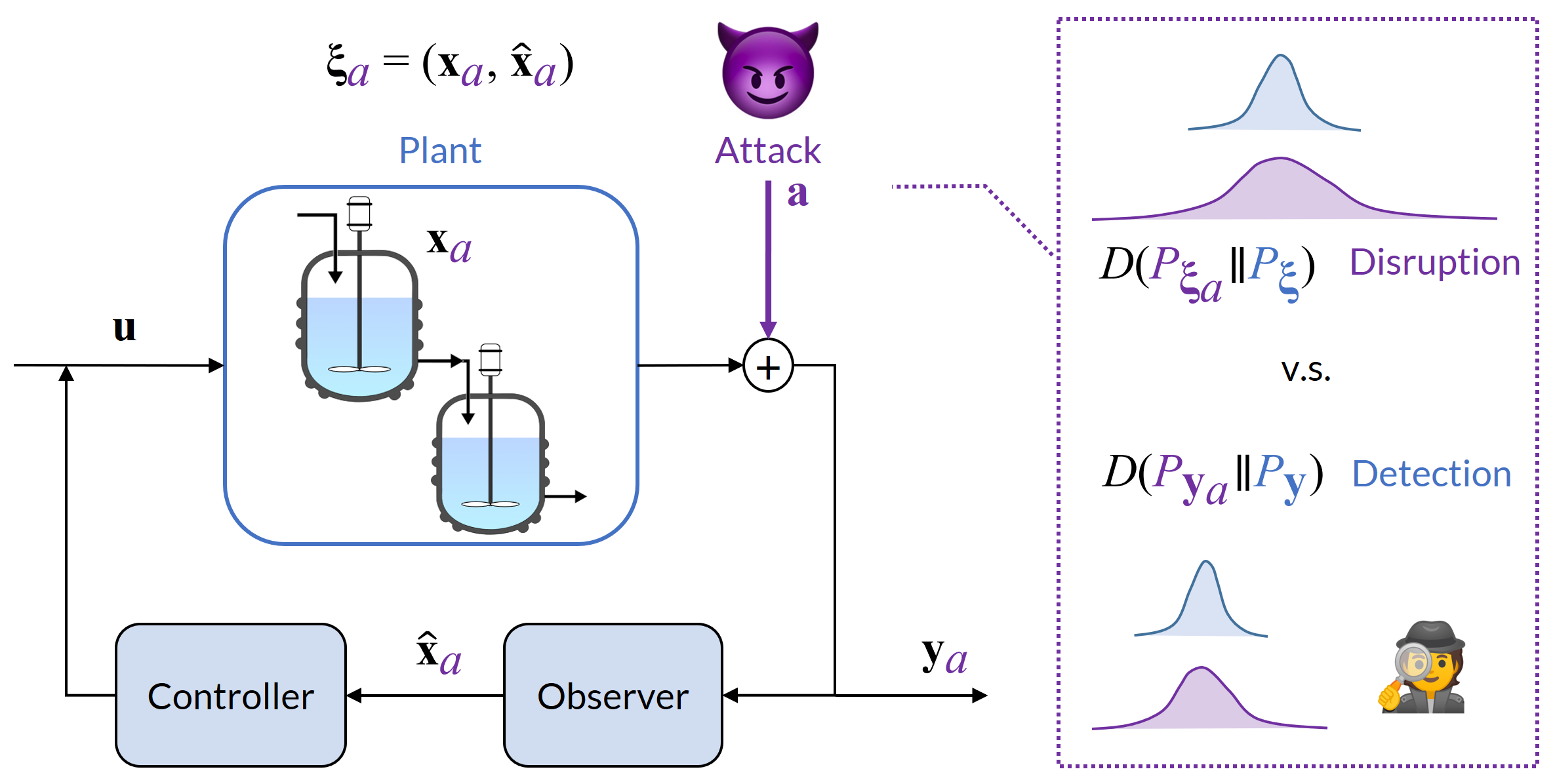}  
\caption{Information-theoretic framework for DIA construction on control systems.}
\label{fig:diag}
\end{figure}

\subsection{Attack Constructions}

Suppose that the attacker seeks for a Gaussian distributed attack $\av \sim {\cal N}(\muv_a, \Sigmam_{\av \av})$. The optimization problem in~\eqref{obj} boils down to the following attack construction problem:
\begin{IEEEeqnarray}{ll}\label{obj_attack} 
\begin{aligned} 
\! \! \min_{{\Sigmam_{\av \av} \in \Sc^m_+, \ \muv_a \in \mathbb{R}^m}} \!\!\! &\!\!\! \!\!\!- D(P_{\xiv_a}   || P_{\xiv}) +\lambda D(P_{\yv_a} \| P_{\yv}),\\
\textnormal{s.t.} & \Upsilonm_{\av \av} = \!\!\sum_{n=0}^{\infty } \! \begin{bmatrix}
   \Am  &  \Bm \Km\\
   \Lm \Cm  & \Am - \Lm \Cm + \Bm \Km
   \end{bmatrix}^n \begin{bmatrix} 
\textnormal{\textbf{0}}   \\ 
\Lm  \\
\end{bmatrix}\! \Sigmam_{\av \av}  [\textnormal{\textbf{0}} \  \ \Lm^{\sf T}] \left(\begin{bmatrix}\!
   \Am \!\! &  \!\!\Bm \Km\\
   \Lm \Cm \!\! & \!\!\Am - \Lm \Cm + \Bm \Km
   \end{bmatrix}^{\sf T}  \right)^n\!\!, \\
   & \Sigmam_{\yv_a \yv_a} - \Sigmam_{\yv \yv} = \Sigmam_{\av \av},
\end{aligned}
\end{IEEEeqnarray}
where $\Upsilonm_{\av \av} \eqdef \Sigmam_{\xiv_a \xiv_a } -  \Sigmam_{\xiv \xiv }$ and the constraints are on the covariance matrices of the probability distributions $P_{\xiv_a}$, $P_{\xiv}$, $P_{\yv_a}$, and $P_{\yv}$, which
come from~\eqref{dis_y},~\eqref{dis_ya},~\eqref{Sigma_xiaxia} and~\eqref{Sigma_xixi}. The infinite sum arises from the explicit solution of the discrete-time algebraic Lyapunov equation~\eqref{Sigma_xiaxia} and~\eqref{Sigma_xixi}.
By elementary operations, the KL divergence between two multivariate Gaussian distributions~$\yv_a \sim \Nc (\muv_a, \Sigmam_{\yv_a \yv_a})$ and $\yv \sim \Nc (\textnormal{\textbf{0}}, \Sigmam_{\yv \yv})$ can be expressed as
	\be\label{KL}
	D(P_{\yv_a} \| P_{\yv}  ) = \dfrac{1}{2} \left(\textnormal{log} \frac{|\Sigmam_{\yv \yv}|}{|\Sigmam_{\yv_a \yv_a}|}- m  +       \textnormal{tr} \left(   \Sigmam_{\yv \yv}^{-1}\Sigmam_{\yv_a \yv_a} \right) +     \textnormal{tr}\left(\Sigmam_{\yv \yv}^{-1}\muv_a \muv_a^{\sf T}     \right)    \right),
	\ee
and similarly
	\be 
	D(P_{\xiv_a} \| P_{\xiv}  ) = \dfrac{1}{2} \left(\textnormal{log} \frac{|\Sigmam_{\xiv \xiv }|}{|\Sigmam_{\xiv_a \xiv_a}|}- 2n  +       \textnormal{tr} \left(   \Sigmam_{\xiv \xiv}^{-1}\Sigmam_{\xiv_a \xiv_a} \right) +     \textnormal{tr}\left(\Sigmam_{\xiv \xiv}^{-1}\muv_{\xiv_a} \muv_{\xiv_a}^{\sf T}     \right)    \right),
	\ee
Therefore, the optimal attack construction can be characterized by the following proposition.
\begin{prop}\label{Prop_obj}
The multivariate Gaussian attack construction $\av \sim {\cal N} ( \muv_a, \Sigmam_{\av \av} )$ that jointly maximizes the disruption captured by $D(P_{\xiv_a} \| P_{\xiv})$ and minimizes the probability of detection characterized by $D(P_{\yv_a} \| P_{\yv})$ in~\eqref{obj} is given by
\begin{IEEEeqnarray}{ll}\label{prop_obj} 
\begin{aligned} 
\min_{{\Sigmam_{\av \av} \in \Sc^m_+, \ \muv_a \in \mathbb{R}^m}} &   \lambda   \left( \textnormal{tr} \left(   \Mm_1  \Sigmam_{\av \av} \right) + \textnormal{tr}\left(\Mm_1 \muv_a \muv_a^{\sf T}     \right)  - \textnormal{log}  \left|  \textbf{I}_m   + \Mm_1 \Sigmam_{\av \av}\right|     \right)   - \textnormal{tr} \left(   \Mm_2 \Sigmam_{\av \av} \right)     - {\textnormal{tr}\left( \Mm_2 \muv_a \muv_a^{\sf T}  \right)}  \\
& + \textnormal{log} \left|\textbf{I}_m + \Mm_2\Sigmam_{\av \av} \right|\\
   \end{aligned}
\end{IEEEeqnarray} 
with $\textnormal{\textbf{0}}$ being a zero matrix with dimension $n$ by $n$ and the matrices $\Mm_1$ and $\Mm_2$ defined as follows
\begin{IEEEeqnarray}{ll}\label{eq_def_M1} 
\begin{aligned} 
& \Mm_1 \eqdef \Sigmam_{\yv \yv}^{-1} \ \textnormal{and}  \\
& \Mm_2 \eqdef [ \textnormal{\textbf{0}} \ \Lm] \sum_{n=0}^{\infty }\left(\begin{bmatrix}
   \Am  &  \Bm \Km\\
   \Lm \Cm  & \Am - \Lm \Cm + \Bm \Km
   \end{bmatrix}^{\sf T} \right)^n \Sigmam_{\xiv \xiv}^{-1} \begin{bmatrix}
   \Am  &  \Bm \Km\\
   \Lm \Cm  & \Am - \Lm \Cm + \Bm \Km
   \end{bmatrix}^n \begin{bmatrix}
   \textnormal{\textbf{0}} \\
   \Lm
   \end{bmatrix}.
   \end{aligned}
\end{IEEEeqnarray} 
\end{prop}
 \begin{proof} 
{Let $\muv_{\xiv_a}$ be the mean vector of $\xiv_a$. Substitute the explicit expression of KL divergence for multivariate Gaussian distributions in~\eqref{KL} into~\eqref{obj_attack}. It follows that the optimization problem in~\eqref{obj_attack} is as follows:}
{\begin{IEEEeqnarray}{ll} 
\begin{aligned} 
 \min_{{\Sigmam_{\av \av} \in \Sc^m_+, \ \muv_a \in \mathbb{R}^m}} &  - \dfrac{1}{2} \left( \textnormal{log} \frac{|\Sigmam_{\xiv \xiv}|}{| \Sigmam_{\xiv \xiv} + \Upsilonm_{\av \av} |}- 2n  + \textnormal{tr} \left(   \Sigmam_{\xiv \xiv}^{-1}(\Sigmam_{\xiv \xiv} + \Upsilonm_{\av \av})\right) + \textnormal{tr}\left(\Sigmam_{\xiv \xiv }^{-1} \muv_{\xiv_a}\muv_{\xiv_a}^{\sf T} \right) \right) \\ & + \lambda \dfrac{1}{2} \left( \textnormal{log} \frac{|\Sigmam_{\yv \yv}|}{|\Sigmam_{\yv \yv} + \Sigmam_{\av \av}|}- m  +       \textnormal{tr} \left(   \Sigmam_{\yv \yv}^{-1}(\Sigmam_{\yv \yv} + \Sigmam_{\av \av})\right) + \textnormal{tr}\left(\Sigmam_{\yv \yv}^{-1}\muv_a \muv_a^{\sf T}     \right)   \right) \\
\textnormal{s.t.} & \ \  \Upsilonm_{\av \av} =  \sum_{n=0}^{\infty } \begin{bmatrix}
   \Am  &  \Bm \Km\\
   \Lm \Cm  & \Am - \Lm \Cm + \Bm \Km
   \end{bmatrix} ^n \begin{bmatrix} 
\textnormal{\textbf{0}}   \\ 
\Lm  \\
\end{bmatrix}\Sigmam_{\av \av} [\textnormal{\textbf{0}} \  \ \ \Lm^{\sf T}] \left(\begin{bmatrix}
   \Am  &  \Bm \Km\\
   \Lm \Cm  & \Am - \Lm \Cm + \Bm \Km
   \end{bmatrix}^{\sf T} \right)^n.
   \end{aligned}
\end{IEEEeqnarray} }
{From~\eqref{20250107_1}, at steady state, the mean vector is $\muv_{\xiv_a} = \left(\textbf{I}_{2n} - \begin{bmatrix}
   \Am  &  \Bm \Km\\
   \Lm \Cm  & \Am - \Lm \Cm + \Bm \Km
   \end{bmatrix}\right)^{-1} \begin{bmatrix}
   \textbf{0} \\
   \Lm \muv_a
   \end{bmatrix}$. From the trace operation, it follows that }
   \begin{IEEEeqnarray}{ll}
\begin{aligned} 
\textnormal{tr}\left(\Sigmam_{\xiv \xiv }^{-1} \muv_{\xiv_a}\muv_{\xiv_a}^{\sf T} \right) = & \textnormal{tr}\left( \Sigmam_{\xiv \xiv}^{-1}   \sum_{n=0}^{\infty } \begin{bmatrix}
   \Am  &  \Bm \Km\\
   \Lm \Cm  & \Am - \Lm \Cm + \Bm \Km
   \end{bmatrix}^n     \begin{bmatrix}
   \textbf{0} \\
   \Lm \muv_a 
   \end{bmatrix}[\textbf{0} \ \muv_a^{\sf T} \Lm^{\sf T}] \left(\begin{bmatrix}
   \Am  &  \Bm \Km\\
   \Lm \Cm  & \Am - \Lm \Cm + \Bm \Km
   \end{bmatrix}^n \right)^{\sf T} \right). \\
   = & \textnormal{tr}\left( [\textbf{0} \  \Lm^{\sf T}] \sum_{n=0}^{\infty }\left(\begin{bmatrix}
   \Am  &  \Bm \Km\\
   \Lm \Cm  & \Am - \Lm \Cm + \Bm \Km
   \end{bmatrix}^n \right)^{\sf T} \Sigmam_{\xiv \xiv}^{-1}    \begin{bmatrix}
   \Am  &  \Bm \Km\\
   \Lm \Cm  & \Am - \Lm \Cm + \Bm \Km
   \end{bmatrix}^n     \begin{bmatrix}
   \textbf{0} \\
   \Lm 
   \end{bmatrix} \muv_a  \muv_a^{\sf T}   \right). \\
   \end{aligned}
\end{IEEEeqnarray} 
Hence, 
$\textnormal{tr}\left(\Sigmam_{\xiv \xiv }^{-1} \muv_{\xiv_a}\muv_{\xiv_a}^{\sf T} \right) = \textnormal{tr}\left( \Mm_2  \muv_a\muv_a^{\sf T} \right)$. 
Similarly,
$\textnormal{tr}\left( \Sigmam_{\xiv \xiv}^{-1} ( \Sigmam_{\xiv \xiv} + \Upsilonm_{\av \av})\right) = 2n + \textnormal{tr}\left(\Sigmam_{\xiv \xiv }^{-1} \Upsilonm_{\av \av}\right) = 2n + \textnormal{tr}\left(\Mm_2\Sigmam_{\av\av}\right)$ and $\textnormal{log} | \textbf{I}_{2n} + \Sigmam_{\xiv \xiv}^{-1} \Upsilonm_{\av \av} | = \textnormal{log} | \textbf{I}_m + \Mm_2 \Sigmam_{\av \av}|$.
The proof is completed by removing the constant terms and a common factor.
\end{proof}
{In the attack construction, we assume that the attacker knows the system model and the second-order moments of the sensor measurements as well as the joint vector of states and state estimates, i.e., $\Sigmam_{\yv \yv}$ and $\Sigmam_{\xiv \xiv}$. Thus, the attacker seeks to optimize both the mean vector and the covariance matrix in~\eqref{prop_obj}.  
From a practical point of view, making the system models and the historical data available to the attacker poses a security threat, and in fact, due to practical and operational constraints, it is more reasonable to assume that the attacker has limited knowledge of the system. However, it is generally difficult to model how much the attacker may know about the system. Hence, as a conservative approach, we assume that the attacker has full knowledge on the system model and the second-order moments. In the next section, we assume the attacker has full access to the measurements of the system, which yields a {\it full attack}. This setting serves as a most conservative approach for the system operators to study possible effects of DIAs.
Following with full attack, we restrict the access to the measurements for the attacker but still assume full knowledge of the model and the second-order moments, which poses a sparsity constraint on attack constructions and yields {\it sparse attack}.}

\section{Full attack}\label{sec_fullattack}
In this section, we consider the attack construction where the attacker has access to all the sensor measurements on the system, which yields a {\it full attack}. In other words, there is no sparsity constraints on $\muv_a \in \mathbb{R}^m$ and $\Sigmam_{\av \av} \in \Sc^m_+$.
In this setting, the attack construction is to solve the optimization problem in~\eqref{prop_obj} as discussed in the last section. First, note that $\mathbb{R}^m$ and $\Sc_+^m$ are both convex sets. For the mean vector, the following proposition characterizes the condition for the convexity of the cost function with respective to $\muv_a$ and gives an expression of the optimal solution.
\begin{prop}\label{Prop_mu}
Suppose that 
\be\label{assp_lambda_mu}
\lambda \Sigmam_{\yv \yv}^{-1}       \succ \Mm_2,
\ee
where $\Mm_2$ and $\Sigmam_{\yv \yv}$ are as in~\eqref{eq_def_M1} and~\eqref{dis_y}, respectively. Then, the cost function in~\eqref{prop_obj} is convex in $\muv_a$ and the optimal solution for $\muv_a$ is 
\be\label{opt_mu}
\muv_a^* = \textnormal{\textbf{0}}.
\ee
\end{prop}
\begin{proof}
From~\eqref{prop_obj}, by removing the constants, the optimal attack construction $\av \sim {\cal N}(\muv_a, \Sigmam_{\av \av})$ with respect to $\muv_a$ is determined by
\begin{IEEEeqnarray}{ll}\label{obj_mu} 
\min_{ \muv_a \in \mathbb{R}^m}   \lambda   \textnormal{tr}\left(\Sigmam_{\yv \yv}^{-1}\muv_a \muv_a^{\sf T}     \right)    -  \textnormal{tr}\left(  \Mm_2 \muv_a \muv_a^{\sf T} \right).
\end{IEEEeqnarray} 
Considering the property of the trace operation $\textnormal{tr}\left(\Sigmam_{\yv \yv}^{-1} \muv_a \muv_a^{\sf T}     \right) = \muv_a^{\sf T}\Sigmam_{\yv \yv}^{-1} \muv_a$, the minimization problem in~\eqref{obj_mu} is equivalent to 
\begin{IEEEeqnarray}{ll}\label{opt_meanvector} 
\begin{aligned} 
\min_{ \muv_a \in \mathbb{R}^m}  \muv_a^{\sf T} \left( \lambda    \Sigmam_{\yv \yv}^{-1} -   \Mm_2 \right) \muv_a.
\end{aligned}
\end{IEEEeqnarray} 
The proof is completed by noting that the assumption in~\eqref{assp_lambda_mu} implies the positive definiteness of the matrix in $\lambda \Sigmam_{\yv \yv}^{-1} - \Mm_2$.
\end{proof}
{The idea of Proposition~\ref{Prop_mu} is that when there is no constraint on the mean vector $\muv_a$ of the attack $\av$ in the full attack, the optimal mean vector is $\muv_a^* =\textbf{0}$ under the condition $\lambda \Sigmam_{\yv \yv}^{-1} \succ \Mm_2$. In other words, if the attacker has full access to all the sensor measurements on the system, the optimal attack construction is {\it centered}. This is naturally expected
as in this case, $D(P_{\yv_a} \| P_{\yv})$ measures the distance from the distribution $P_{\yv_a}$ to $P_{\yv}$. Given that the mean vector of $P_{\yv}$ is a zero vector as in~\eqref{dis_y}, the optimal mean vector of $P_{\yv_a}$ to minimize $D(P_{\yv_a} \| P_{\yv})$ needs to be a zero vector as well. }
However, the attacker also considers the attack effect on the stationary distributions $D(P_{\xiv_a} \| P_{\xiv })$ with a weighting parameter $\lambda$. {The terms in~\eqref{obj_mu} that are associated with $D(P_{\yv_a} \| P_{\yv})$ and $D(P_{\xiv_a} \| P_{\xiv })$ are $\textnormal{tr}\left(\Sigmam_{\yv \yv}^{-1}\muv_a \muv_a^{\sf T} \right)$ and $\textnormal{tr}\left(\Mm_2 \muv_a \muv_a^{\sf T} \right)$, respectively. } 
The assumption in~\eqref{assp_lambda_mu} in this proposition implies a lower bound of $\lambda$ to overcome the non-convexity introduced by the $ - D(P_{\xiv_a} \| P_{\xiv })$ term and 
make the cost function in~\eqref{opt_meanvector} convex in $\muv_a$. Hence, the weighting parameter $\lambda$ that governs the tradeoff between $D(P_{\yv_a} \| P_{\yv})$ and $D(P_{\xiv_a} \| P_{\xiv })$ is constrained by a lower bound, which means that the attacker needs to consider attack detection to a certain minimum extent.
It is worth noting that the assumption is associated with the system model, covariance matrix of the measurements, the observer gain and the covariance of the dynamics of the states and estimates in the system. 

{As the attack vector is $\av  \sim {\cal N}(\muv_a, \Sigmam_{\av \av})$, Proposition~\ref{Prop_mu} provides the optimal solution for $\muv_a$. We now proceed with the optimization with respect to the covariance matrix $\Sigmam_{\av \av}$ in~\eqref{prop_obj}. Replacing $\muv_a$ with the optimal mean vector $\muv_a^* = \textbf{0}$, from Proposition~\ref{Prop_obj}, the attack construction $\av \sim {\cal N} ( \textnormal{\textbf{0}}, \Sigmam_{\av \av} )$ is now cast as:}
\begin{IEEEeqnarray}{ll}\label{1114_01} 
\begin{aligned} 
\min_{\Sigmam_{\av \av} \in \Sc^m_+ } &   \lambda   \left( \textnormal{tr} \left(   \Mm_1  \Sigmam_{\av \av} \right)    - \textnormal{log}  \left|  \textbf{I}_m   + \Mm_1 \Sigmam_{\av \av}\right|     \right)   - \textnormal{tr} \left(   \Mm_2 \Sigmam_{\av \av} \right)    + \textnormal{log} \left|\textbf{I}_m + \Mm_2\Sigmam_{\av \av} \right|.
   \end{aligned}
\end{IEEEeqnarray} 
{We denote the cost function in~\eqref{1114_01} as $f(\Sigmam_{\av \av})$ and examine the first-order derivatives and the Hessian.} Given the symmetry of $\Sigmam_{\av \av}$ and the matrix differential~\cite{seber}, after elementary operations, we find
\be\label{derivative_fullattack} 
\begin{aligned}
\frac{d f}{d \Sigmam_{\av \av}}= \Sigmam_{\av \av}( \Mm_2- \lambda \Mm_1) \Sigmam_{\av \av} + \Mm_1^{-1} \Mm_2 \Sigmam_{\av \av}    - \lambda \Sigmam_{\av \av}\Mm_1 \Mm_2^{-1},
\end{aligned}
\ee
and after the vectorization of $\Sigmam_{\av \av}$, the Hessian is then given as follows:
\be\label{Hessian}
\frac{\partial f^2}{\partial \textnormal{vec}(\Sigmam_{\av \av})\partial \textnormal{vec}(\Sigmam_{\av \av})} = 2\Sigmam_{\av \av} \otimes (\Mm_2 -\lambda \Mm_1) + \textbf{I}_m \otimes ( \Mm_2 \Mm_1^{-1} + \Mm_1^{-1}\Mm_2) - \lambda( \Mm_1 \Mm_2^{-1} + \Mm_2^{-1}\Mm_1 ) \otimes \textbf{I}_m.
\ee
To establish the convexity of the cost function and the corresponding optimal solution for the minimization problem in~\eqref{1114_01}, we propose the following theorem. The rationale of this theorem is that the first-order derivative in~\eqref{derivative_fullattack} gives the stationary point and a positive definite Hessian matrix in~\eqref{Hessian} can guarantee the convexity of the cost function. 
Hence, the attacker can construct the optimal attack accordingly. {To that end, the following theorem provides the convexity conditions for the cost function $f(\Sigmam_{\av \av})$ and the optimal solution $\Sigmam_{\av \av}$ of the optimization problem in~\eqref{1114_01}, i.e., the optimal covariance matrix of the attack vector.}
\begin{theorem}\label{Theorem_fullattack}
Suppose that (i) all eigenvalues of $ (\Mm_1^{-1}\Mm_2)^{\sf T} -\lambda \Mm_1\Mm_2^{-1}$ are positive, (ii) all eigenvalues of $ \Mm_2 - \lambda \Mm_1 $ are negative, and furthermore (iii) $\lambda$ satisfies the following condition  
    \be\label{condition_second_fullattack}
    \lambda < \frac{\lambda_{\min} \left(2\Sigmam_{\av \av} \otimes\Mm_2 \right)}{\lambda_{\max}\left( 2\Sigmam_{\av \av} \otimes\Mm_1 + (\Mm_1\Mm_2^{-1}  + \Mm_2^{-1}\Mm_1) \otimes \textbf{I}_m \right)},
    \ee 
    where $\Sigmam_{\av \av}$ is as follows:
    \begin{IEEEeqnarray}{ll}\label{Sigma_aa_inv}
    \begin{aligned} 
    \textnormal{vec} ( \Sigmam_{\av \av}^{-1}) = &- (( \Qm + \Pm^{\sf T} ) \otimes \textbf{I} + \textbf{I} \otimes (\Pm^{\sf T} + \Qm) )^{-1}\textnormal{vec} (\Sm + \Sm^{\sf T} ),
    \end{aligned}
    \end{IEEEeqnarray} 
    with $\Sm \eqdef \Mm_2 - \lambda \Mm_1$, $\Pm \eqdef \Mm_1^{-1} \Mm_2$, and $\Qm \eqdef -\lambda \Mm_1\Mm_2^{-1}$. Then, $\Sigmam_{\av\av}$ such that~\eqref{Sigma_aa_inv} holds is unique and optimal.
\end{theorem}
\begin{proof}
Symmetrizing~\eqref{derivative_fullattack} yields
$\Sigmam_{\av \av} (\Sm + \Sm^{\sf T}) \Sigmam_{\av \av} + \Sigmam_{\av \av} (\Qm + \Pm^{\sf T}  ) + (\Pm + \Qm^{\sf T}) \Sigmam_{\av \av}   = \textnormal{\textbf{0}}.$
Under the given conditions $ - (\Pm^{\sf T} + \Qm )$ is Hurwitz and $\Sm + \Sm^{\sf T} \prec \textnormal{\textbf{0}}$. Hence, there exists a unique positive definite solution $X$ such that
\begin{IEEEeqnarray}{ll}\label{1107_1}
\begin{aligned} 
-(\Qm + \Pm^{\sf T}  )X - X  (\Pm + \Qm^{\sf T}) = & \Sm + \Sm^{\sf T}.
\end{aligned}
\end{IEEEeqnarray} 
The existence and uniqueness are guaranteed by the properties of continuous-time algebraic Lyapunov equations~\cite{seber}. The solution, given on the right-hand-side of~\eqref{Sigma_aa_inv}, is a stationary point. The convexity is guaranteed if and only if the Hessian in~\eqref{Hessian} is positive semidefinite in the domain, for which a sufficient condition is given in~\eqref{condition_second_fullattack}.
The proof is completed.
\end{proof}  
Theorem~\ref{Theorem_fullattack} provides the analytical solution for the full attack constructions and the conditions to guarantee the uniqueness and optimality of the solution. 
{Note that without convexity, a stationary point is not an optimal solution to the attack construction. The optimality only holds when the condition in~\eqref{condition_second_fullattack} is satisfied, and hence attacker must choose $\lambda$ carefully such that the condition is satisfied to get an optimal attack. Apart from the upper bound of $\lambda$ discussed in this Theorem, the attacker needs to use a $\lambda$ that conforms the lower bound of $\lambda$ as in~\eqref{assp_lambda_mu} in Proposition~\ref{Prop_mu}. Both the upper and lower bounds are associated with system model and second-order moments of process variables. In fact, there exist cases that the attacker can not find a feasible $\lambda$ satisfying both conditions. Consequently, it is not guaranteed that such an optimal full attack always exists.} 

The attack constructions in~\eqref{Sigma_aa_inv} and~\eqref{opt_mu} yield an attack vector where all the entries of the attack realizations are allowed to be nonzero, which means that the attack implementation requires full access to all the measurements $\yv$ of the systems. In practice, DIAs may intrude only a subset of the sensing infrastructures due the vulnerabilities existing in a local data acquisition mechanism. For that reason, studying full attacks only provides a baseline for DIAs on control systems. It is the study of sparse attacks that restrict the attacker's access to a limited amount of measurements that is of particular interest. 

\section{Sparse Attack}\label{sec_kattack}
In this section, we study sparse attacks where the attacker seeks for the attacks over a subset of measurements on the system, i.e., we pose a $k$-sparsity constraint on attack vector. {In the previous section, the attacker has full access to all the sensor measurements in the system, whereas in this section, the attacker has access to only $ k < m$ sensor measurements. }Specifically, we set the feasible domain in~\eqref{obj_attack} as the set of a $m$-dimensional Gaussian distributions that place nonzero masses on $k < m$ entries on the attack vector $\av$, which yields a {\it $k$-sparse attack}. The $k$-sparsity constraint is formulated as $|\textnormal{supp}(\av)|_0 = k$, where $\textnormal{supp}(\av )$ denotes the support of the vector $\av$, that is, $\textnormal{supp}(\av ) = \{i: \av_i \neq 0\}$. Following the conclusion in Proposition~\ref{Prop_mu}, we assume $\lambda  \Mm_2^{-1} \succ \Sigmam_{\yv \yv}$ and obtain a zero vector as the optimal mean vector for the $k$-sparse attack construction. Therein, the attacker only needs to consider the sparsity constraint on the covariance matrix of the attack vector $\Sigmam_{\av \av}$, which now translates into a constraint on the number of nonzero entries in the diagonal of $\Sigmam_{\av \av}$. Hence, the feasible domain of $\Sigmam_{\av \av}$ becomes:
\be
\label{eq_def_Sk}
\Sc_k\eqdef \left\{\Wm\in \Sc^m_+: | \textnormal{supp}(\textnormal{diag}(\Wm) ) |_0 =k \right\},
\ee
where $\textnormal{diag}(\Wm)$ denotes the vector formed by the diagonal entries of $\Wm$.
More specifically, the set in~\eqref{eq_def_Sk} denotes a set of positive semidefinite matrices with only $k$ nonzero entries in the diagonal. {In other words, the attacker is looking for a covariance matrix $\Sigmam_{\av \av}$ with $k$ nonzero entries in the diagonal, which correspond to the measurements that the attacker needs to attack, and the nonzero entries are the attack variances on the corresponding measurements.}
{Consequently, the $k$-sparse attack construction is cast as}
\begin{IEEEeqnarray}{ll}\label{att_k}
\begin{aligned}
 \min_{   {\Sigmam_{\av \av} \in \Sc_k}}  &  \lambda \left( \textnormal{tr} \left(   \Mm_1  \Sigmam_{\av \av} \right) - \textnormal{log}  \left|\textbf{I} + \Mm_1 \Sigmam_{\av \av}\right| \right)   -  \textnormal{tr} \left(    \Mm_2 \Sigmam_{\av \av} \right)  + \textnormal{log} \left|\textbf{I}_{m} +  \Mm_2 \Sigmam_{\av \av} \right|,  
\end{aligned}
\end{IEEEeqnarray} 
where the new feasible domain $\Sc_k$ is as in~\eqref{eq_def_Sk}. The feasible domain $\Sc_k$ naturally leads to a combinatorial nature of the problem, where $k$ measurements are selected from the full set of the measurements in the system and the corresponding variance are determined. Therefore, a recursive or separable substructure that simplifies the search for the optimal solution is essential to this problem. {In this paper, we tackle the combinatorial nature of this problem by breaking down the measurement selection problem in a one-at-a-time fashion. This leads to single measurement attacks in a sequential process. In the single measurement attack case, the attack construction boils down to assessing the optimal attack for each measurement and comparing the attack performance among all the measurements. }

\subsection{Optimal Single Measurement Attacks}\label{sec_singleattack}
In the single measurement attack construction, the domain in~\eqref{eq_def_Sk} is narrowed down to the set of matrices where only one nonzero entry is contained in the diagonal and all the other entries are zeros, which we denote by $\Sc_1$. {This nonzero entry corresponds to the attack variance on the single measurement that is attacked}. Hence, the  minimization problem in~\eqref{att_k} is cast as follows:
\begin{IEEEeqnarray}{ll}\label{1120_1} 
\begin{aligned}
 \min_{   {\Sigmam_{\av \av} \in \Sc_1}}  &  \lambda   \left( \textnormal{tr} \left(   \Mm_1  \Sigmam_{\av \av} \right)    - \textnormal{log}  \left|\textbf{I} + \Mm_2 \Sigmam_{\av \av}\right|     \right) - \textnormal{tr} \left(    \Mm_2 \Sigmam_{\av \av} \right)  
 + \textnormal{log} \left|\textbf{I}_{m} +  \Mm_2 \Sigmam_{\av \av} \right|.
\end{aligned}
\end{IEEEeqnarray} 
{To solve the problem above, we propose the following theorem that provides a condition on $\lambda$ to guarantee the existence of an optimal solution to the single measurement attack construction, as well as the explicit form of the solution. }
\begin{theorem}
Let $i \in \{1,2,...,m\}$ be the index of the attacked sensor such that $\ev_i^{\sf T}\Sigmam_{\av \av} \ev_i = v > 0 $, namely, the variance of the single measurement attack is positive.
Denote matrices $\Mm_1$ and $\Mm_2 $ as in~\eqref{eq_def_M1}. Suppose that the weighting parameter $\lambda$ satisfies
\be\label{singleattack_condition_lambda}
\frac{\ev_i^{\sf T} \Mm_2\ev_i}{\ev_i^{\sf T} \Mm_1\ev_i} < \lambda < \frac{(\ev_i^{\sf T} \Mm_2\ev_i)^2}{(\ev_i^{\sf T} \Mm_1\ev_i)^2}.
\ee
Then, the optimal solution to the single measurement attack problem in~\eqref{1120_1} uniquely exists, and is specified by
\begin{subequations}\label{singleattack_sol}
\begin{align}
v^*_i = & \frac{ (\ev_i^{\sf T} \Mm_2 \ev_i)^2 -\lambda (\ev_i^{\sf T} \Mm_1\ev_i)^2  }{ (\lambda \ev_i^{\sf T} \Mm_1\ev_i - \ev_i^{\sf T} \Mm_2 \ev_i) \ev_i^{\sf T} \Mm_1\ev_i \ev_i^{\sf T} \Mm_2 \ev_i    },\\
i^* = & \min_i \lambda   \left( \textnormal{tr} \left(   \Mm_1 v_i^* \ev_i \ev_i^{\sf T} \right)  - \textnormal{log}  \left|\textbf{I} + \Mm_1 v_i^* \ev_i \ev_i^{\sf T}\right|     \right)    -  \textnormal{tr} \left(    \Mm_2 v_i^* \ev_i \ev_i^{\sf T} \right)  + \textnormal{log} \left|\textbf{I}_{m} +  \Mm_2 v_i^* \ev_i \ev_i^{\sf T} \right|.
\end{align}
\end{subequations}
\end{theorem}
\begin{proof}
Note that $\Sigmam_{\av \av} = v\ev_i \ev_i^{\sf T}$. 
The single measurement attack in~\eqref{1120_1} is equivalent to
\begin{IEEEeqnarray}{ll}\label{20250320_1} 
\begin{aligned}
\min_{i} \min_{v \in \mathbb{R}^+} \lambda \ev_i^{\sf T} \Mm_1\ev_i v - \lambda \textnormal{log}(1+\ev_i^{\sf T} \Mm_1\ev_i v) - \ev_i^{\sf T} \Mm_2\ev_i v + \textnormal{log} (1+\ev_i^{\sf T} \Mm_2\ev_iv).
\end{aligned}
\end{IEEEeqnarray} 
The resulting single measurement attack can then be solved as an inner minimization and an outer minimization problem. {For a specific measurement $i$, the inner minimizer is actually the optimal attack variance $v$ to this measurement. After getting all the optimal variances for all $i \in \{1,2,...,m\}$, the optimal single measurement attack is obtained by comparing the performance of each measurement in terms of the cost function, which is characterized by the outer minimization. }
Since the derivative of the cost function in the inner minimizer $v$ is quadratic, the conditions such that cost function is convex in $v$ and the solution $v$ is positive real are given by
\begin{equation}\label{single_attack_condition}
\begin{aligned}
\begin{cases}
(\lambda \ev_i^{\sf T} \Mm_1\ev_i - \ev_i^{\sf T} \Mm_2\ev_i)(\ev_i^{\sf T} \Mm_1\ev_i)(\ev_i^{\sf T} \Mm_2\ev_i) > 0,\\
v =  \frac{ (\ev_i^{\sf T} \Mm_2\ev_i)^2 -\lambda (\ev_i^{\sf T} \Mm_1\ev_i)^2  }{ (\lambda (\ev_i^{\sf T} \Mm_1\ev_i) - \ev_i^{\sf T} \Mm_2\ev_i) (\ev_i^{\sf T} \Mm_1\ev_i)(\ev_i^{\sf T} \Mm_2\ev_i)   }  > 0.
\end{cases}
\end{aligned}
\end{equation}
These result in the condition in~\eqref{singleattack_condition_lambda} and the optimal inner minimizer in~\eqref{singleattack_sol}. The outer minimizer is then obtained by evaluating the optimal attack on each measurement $i$. This completes the proof.
\end{proof} 
\subsection{$k$-sparse Attacks}
The attack construction in the last section provides the optimal solution to the single measurement attack. To circumvent the combinatorial nature of seeking $k$ measurements simultaneously, we propose a greedy construction that sequentially selects one-at-a-time measurement at each step. {At each step, the problem is to select one measurement to attack, which has the same problem structure of single measurement attack as in the previous subsection. Hence, the idea of breaking the single measurement attack into inner and outer optimization in single measurement attack construction is adopted at each step of the sequential process in $k$-sparse attacks. } 

We assume that the entries of the attack vector are independent, i.e., 
$P_{\av}=\prod_{i=1}^k P_{\av_i}$, 
{where, all $P_{{\av}_i}$, $i \in \lbrace 1,2, \ldots, k\rbrace$, are Gaussian with zero mean and variance $v_i$, that is, $\av_i \sim {\cal N}(0, v_i)$.}
As a result, the set of covariance matrices given by~\eqref{att_k}, with $k < m$, is narrowed down to the set 
\be\label{20250105_1}
\bar{\Sc}_k \eqdef \bigcup_{\Kc}\!\left\{\Sm\!\in\! \Sc^m_+\!\!:\Sm\!=\!\!\sum_{i\in\Kc}v_i \ev_i\ev_i^{\sf T}\textnormal{with}\;v_i\!\in\!\mathbb{R}_+\!\right\},
\ee
where $|\Kc|_0=k $. The optimization problem in~\eqref{1114_01} boils down to the following problem:
\be
\label{eq:Gaussian_k_stealth_indep}
\min_{{\Sigmam}_{\av \av} \in \bar{\Sc}_k } f(\Sigmam_{\av \av}),
\ee
for which we sequentially update the index of one more nonzero entry and the corresponding value in the diagonal of $\Sigmam_{\av \av}$, {which corresponds to a new selected sensor measurement to attack and the attack variance to that measurement.}
Specifically, given the sparsity constraint in \eqref{20250105_1}, the construction is completed through $k$ steps. At each step, a new index of the measurement is added to $\Ac$. 
At step $i$, let $\Sigmam_i \in \Sc_{+}^m$ be the covariance matrix of the vector attack under construction, and $\mathcal{A}_{i}$ be the set of indices corresponding to the entries of the vector $\textnormal{diag}(\Sigmam_i)$ that are nonzero, {which is actually the set of measurements on the system that are under attacks until step $i$.}
That is,
\begin{equation}
\label{EqSetAi}
\mathcal{A}_i = \lbrace j \in \lbrace 1,2, \ldots, m \rbrace:  \ev_j^{\sf T}  \Sigmam_i \ev_j > 0 \rbrace.
\end{equation}
{The complement set of $\mathcal{A}_{i}$, denoted by $\mathcal{A}_{i-1}^{\sf c}$, is the set of measurements on the system that are not under attacks until step $i$, namely $\Ac^c_i = \{1, 2, \dots, m\}  \setminus \mathcal{A}_i$}.
Therefore, $\mathcal{A}_{i} \subseteq \lbrace 1,2, \ldots, m \rbrace$ and has a cardinality of $ i$ and $\mathcal{A}_{1} \subset \mathcal{A}_{2} \subset \ldots \subset \mathcal{A}_{k} \subset  \lbrace 1,2, \ldots, m \rbrace$.
Hence, the selection procedure can be written as
\begin{IEEEeqnarray}{ll} 
\begin{aligned} 
\label{eq:greed_step}
\Sigmam_{i} = \Sigmam_{i-1}+ v \ev_j \ev^{\sf T}_j, \ \ 
\Ac_i = \Ac_{i-1} \cup \{j\},
\end{aligned}
\end{IEEEeqnarray} 
where $v \in \mathbb{R}_+$ is the value of the new nonzero entry, which is the optimal variance to attack the measurement $j$, and $j \in \Ac_{i-1}^c$ is the identified measurement at step $i$.

Determining both $j \in \mathcal{A}_{i-1}^{\sf c}$ and $v \in \mathbb{R}_+$ at step $i$ as described in~\eqref{eq:greed_step} is based on the solution to the following optimization problem:
		\be
		\label{op:indep_sel}
		\min_{(j, v)\in\mathcal{A}_{i-1}^{\sf c}\times\mathbb{R}_+} f(\Sigmam_{i-1}+v \ev_j\ev^{\sf T}_j).
		\ee
{The structure of this optimization problem is the same as the single measurement attack in~\eqref{20250320_1}, which was treated as an inner and outer optimization in single measurement attack construction. Before we proceed with the decomposition into an inner optimization problem and an outer optimization problem, we propose the following lemma to enable the selection of both $j \in\mathcal{A}_{i-1}^{\sf c}$ and $v > 0$ at step $i$ based on a simpler but equivalent optimization problem to~\eqref{op:indep_sel}.} 
\begin{lemma}
\label{lm:cost_diff}
Let $\Sigmam_1\in \Sc^m_+$ and $\Sigmam_2\in \Sc^m_+$ be two matrices that satisfy $\Sigmam_2=\Sigmam_1+\bm{\Delta}$, with $\bm{\Delta}\in\mathbb{R}^{m\times m}$. Then, the cost function $f$ defined in \eqref{op:indep_sel} satisfies that 
\be\label{20250105_2}
f(\Sigmam_2) - f(\Sigmam_1) = g( \bm{\Delta}),
\ee
where
\be\label{Eqf}
\begin{aligned}
g( \Deltam) \eqdef &  \lambda   \left( \textnormal{tr} \left(   \Mm_1  \Deltam \right)  - \textnormal{log}  \left|\textbf{I} + (\Mm_1^{-1} + \Sigmam_1)^{-1}\Deltam \right|     \right)   -  \textnormal{tr} \left(  \Mm_2 \Deltam \right)  + \textnormal{log} \left|\textbf{I} +  ( \Mm_2^{-1} +  \Sigmam_1)^{-1} \Deltam  \right|,  \\
\end{aligned}
\ee
with $\Mm_1$ and $\Mm_2$ defined in~\eqref{eq_def_M1}.
\end{lemma}
\begin{proof}
The conclusion is obtained via elementary calculation. 
\end{proof}
The above lemma sheds light on a simpler optimization problem than~\eqref{op:indep_sel}. {The arguments in~\eqref{op:indep_sel} are $j$ and $v$, which translate into $\Deltam$ in~\eqref{20250105_2}. At the proposed sequential process, at each step, the attacker minimizes the additional cost to attack one more measurement, that is, $f(\Sigmam_2) - f(\Sigmam_1)$. From this lemma, minimizing additional cost is equivalent to minimizing $g(\Deltam)$ in~\eqref{Eqf} for a $\Deltam \in \mathbb{R}^{m \times m}$ of the form $\Deltam = v \ev_j\ev_j^{\sf T}$. Hence, at each step $i$, the attacker needs to solve the following minimization problem:}
{\begin{IEEEeqnarray}{ll}\label{op:indep_sel2}
	\begin{aligned}
		&\min_{(j,v)\in\mathcal{A}_{i-1}^{\sf c} \times\mathbb{R}_+} g(  v \ev_j\ev^{\sf T}_j) \\
		= &\min_{(j,v)\in\mathcal{A}_{i-1}^{\sf c} \times\mathbb{R}_+}     \lambda   \left( \textnormal{tr} \left(   \Mm_1 v\ev_j\ev_j^{\sf T} \right)  - \textnormal{log}  \left|\textbf{I} + (\Mm_1^{-1} + \Sigmam_1)^{-1}v\ev_j\ev_j^{\sf T} \right|     \right)   -  \textnormal{tr} \left(  \Mm_2 v\ev_j\ev_j^{\sf T} \right)  \\
		& \ \ \ \ \ \ \ \ \ \ \ \ \ \ \ \ \ \ \ \  + \textnormal{log} \left|\textbf{I} +  ( \Mm_2^{-1} +  \Sigmam_1)^{-1}v\ev_j\ev_j^{\sf T}  \right|,
\end{aligned}
\end{IEEEeqnarray}}
{The minimization problem at step $i$ is essentially determining which measurement $j \in {\cal A}_{i-1}^c$ to compromise and the corresponding attack variance $v \in \mathbb{R}_+$, which is equivalent to the following two-layer minimization problem:}
\begin{IEEEeqnarray}{ll}\label{20211220_1} 
	\begin{aligned}
	   \min_{j \in\mathcal{A}_{i-1}^{\sf c}}\min_{ v \in \mathbb{R}_+} & \ \ \lambda   \left( v   \ev_j^{\sf T}   \Mm_1  \ev_j - \textnormal{log} (1+v \ev_j^{\sf T}(\Mm_1^{-1} +  \Sigmam_1)^{-1} \ev_j  )      \right)   - v\ev_j^{\sf T} \Mm_2  \ev_j  \\
	  & \ + \textnormal{log}  (1+ v\ev_j^{\sf T}(\Mm_2^{-1} +  \Sigmam_1)^{-1} \ev_j  ).
	\end{aligned}
\end{IEEEeqnarray}	
{In other words, at step $i$, for each measurement $j \in \Ac_i^c$, the attacker seeks for the optimal attack variance $v$ by the inner minimization, and then determines the best measurement to attack by comparing the performance of each measurement with their own optimal attack variance obtained in inner minimization. This strategy is inherited from the single measurement attacks case.}
	
It can be shown that the first-order derivative of the cost function in~\eqref{20211220_1} with respect to $v$ is quadratic multiplied by a strictly positive factor. For a measurement $j$, we need to make assumptions on the weighting parameter $\lambda$ to obtain the optimal solution of $v$. In this setting, the conditions on the weighting parameter is not uniform in $j$. In fact, for each measurement $j$, the conditions on $\lambda$ depends on $\ev_j^{\sf T}   \Mm_1  \ev_j, \ev_j^{\sf T}(\Mm_1^{-1} +  \Sigmam_1)^{-1} \ev_j, \ev_j^{\sf T} \Mm_2  \ev_j$ and $\ev_j^{\sf T}(\Mm_2^{-1} +  \Sigmam_1)^{-1} \ev_j$ from~\eqref{20211220_1}.
Let us denote
{\begin{subequations}\label{abcd} 
	\begin{align} 
		a_j \eqdef \ev_{j}^{\sf T}\Mm_1 \ev_{j}, \ b_j \eqdef \ev_{j}^{\sf T}\Mm_2 \ev_{j}, \ c_j \eqdef  \ev_{j}^{\sf T}(\Mm_1^{-1} +\Sigmam_1)^{-1}\ev_{j}, \ \textnormal{and } \
		d_j \eqdef \ev_{j}^{\sf T}(\Mm_2^{-1} +\Sigmam_1)^{-1}\ev_{j}.
	\end{align}
\end{subequations} 
After elementary operations, to evaluate given the first-order derivative of the cost function in~\eqref{20211220_1}, we find that the convexity with respect to $v \in \mathbb{R}_+$ holds if and only if $\lambda$ satisfies the following inequalities, for all $j \in \Ac^c_{i-1}$:}
{\begin{equation}\label{lambda_constraints}
	\begin{aligned}
		\begin{cases}
		(1)&	\lambda a_j - b_j > 0, \\
		(2)&	((\lambda a_j - b_j)(c_j+d_j) + c_jd_j(1-\lambda))^2 - 4c_jd_j(\lambda a_j - b_j)(d_j - \lambda c_j) > 0,\\
		(3)&	(b_j-\lambda a_j)(c_j+d_j) + c_jd_j(\lambda -1) \\
            &+ \sqrt{((\lambda a_j - b_j)(c_j+d_j) + c_jd_j(1-\lambda))^2 - 4c_jd_j(\lambda a_j - b_j)(d_j - \lambda c_j)} > 0.
		\end{cases}
	\end{aligned}
\end{equation}}
With the condition in~\eqref{lambda_constraints}, the optimal attack variance to measurement $j$ is denoted by
\begin{equation}\label{v_star}
\begin{aligned}
v_j^* = &\frac{(b_j-\lambda a_j)(c_j+d_j)+c_j d_j(\lambda -1)}{2 c_j d_j(\lambda a_j -b_j) } \\
&+\frac{\sqrt{\left((\lambda a_j - b_j)(c_j+d_j)+c_jd_j(1-\lambda)\right)^2 - 4 c_j d_j(\lambda a_j -b_j)(d_j-\lambda c_j)}}{2 c_j d_j(\lambda a_j -b_j) }.
\end{aligned}
\end{equation}
We remark that the optimal attack variance for the measurement $j$ actually depends only on the weighting parameter $\lambda$ in~\eqref{lambda_constraints}, since $a_j, b_j, c_j, d_j$ are constants. The weighting parameter $\lambda$ determines the convexity and optimality for the attack variance. {This is expected, since given the attacks on the system at this step, for a specific measurement $j \in {\cal A}_{i-1}^c$, all the parameters in~\eqref{20211220_1} associated with the system and previous attacks are known, and hence the attack variance $v_j$ is the only variable. Given the optimal attack variance for measurement $j$ as in~\eqref{v_star}, the following theorem provides the strategy to select the optimal measurement $j^*$ to attack. } 
\begin{theorem}\label{th:indep_sel}
Let $k$ satisfy $0 < k < m$, $i$ satisfy $0<i<k$, and $\left(j^{\star}, v^{\star} \right) \in \mathcal{A}_{i-1}^{\sf c}\times\mathbb{R}_+$ be
the solution to the optimization problem in~\eqref{op:indep_sel} at step $i$. Let $a_j, b_j, c_j, d_j$ be as defined in~\eqref{abcd}.
Assume that the weighting parameter satisfies~\eqref{lambda_constraints} for all $j \in \Ac^c_i$. Then, at step $i$ in the sequential process, the optimal measurement to attack is 
\be\label{j_star} 
j^* = \argmin_{j \in \Ac_{i-1}^c} f(\Sigmam_{i-1}+v_j^* \ev_{j}\ev^{\sf T}_{j})
\ee
where the optimal attack variance for all $j \in \Ac_{i-1}^c$ is given by~\eqref{v_star}.
\end{theorem}
\begin{proof}
The proof follows by the previous discussion.
\end{proof}

Based on the sequential process and Theorem~\ref{th:indep_sel}, the proposed $k$-sparse attack construction is described in Algorithm~\ref{alg:greedy}. 
\begin{algorithm}[ht]
	\caption{$k$-sparse attack construction}
	\label{alg:greedy}
	\KwIn{System model $(\Am, \Bm, \Cm)$; observer gain $\Lm$; controller gain $\Km$;
	       system noise $\Sigmam_{\nv \nv}$; disturbance $\Sigmam_{\dv \dv}$;
	       sparsity constraint $k$.}
	\KwOut{$\bm{\Sigma}_{\av \av} \in \bar{\Sc}_k$}
	Initialize $\Ac_0 \gets \emptyset$, $\Sigmam_0 \gets \mathbf{0}$\; \\
	\For{$i = 1$ \KwTo $k$}{
		\For {$j \in \Ac_{i-1}^{\sf c}$}{
		Compute the optimal $v_{j}$ as in~\eqref{v_star}\; \\
            Compute the cost function $f$ when attacking measurement $j$ with attack variance $v_j$ defined in~\eqref{op:indep_sel} \;
		}
		Compute $j^{\star}$ in~\eqref{j_star}\;\\
		Set $\Ac_j \gets \Ac_{i-1} \cup \{ j^{\star} \}$\;\\
		Set $\Sigmam_i \gets \sum_{j \in \Ac_j} v_{j} \ev_j \ev_j^{\sf T}$\;
	}
	$\bm{\Sigma}_{\av \av} \gets \Sigmam_k$\;
\end{algorithm}

\section{Numerical Case Study and Vulnerability Analysis}\label{sec_numerical_results}
In this section, we present the simulation results on a chemical process consisting of two continuously stirred tank reactors (CSTRs) in series. The reactant A is fed into the reactors $j \in \{1,2\}$, with inlet concentrations $C_{Aj0}$, inlet temperature $T_{j0}$, and flow rate $F_{j0}$. The heating jacket is installed around reactor $j$ with a manipulable heating / cooling rate $Q_j$.  
The dynamics of this two-CSTR-in-series process is described as follows~\cite{CW_CERD_2021}:
\begin{subequations} 
\begin{align} 
  \frac{d C_{A1}}{d t} = & \frac{F_{10}}{V_{L_1}} (C_{A10} - C_{A1}) -k_0 e^{\frac{-E}{RT_1}}C_{A1}^2, \\
  \frac{d T_1}{d t} = & \frac{F_{10}}{V_{L1}}(T_{10} - T_1) - \frac{ \Delta H}{\rho_L C_p}k_0 e^{\frac{-E}{RT_1}}C_{A1}^2 + \frac{Q_1}{\rho_L C_p V_{L_1}},\\
  \frac{d C_{A2}}{d t} = & \frac{F_{20}}{V_{L_2}} C_{A20} + \frac{F_{10}}{V_{L_2}}  C_{A1 } - \frac{F_{10}+F_{20}}{V_{L_2}} C_{A2} -k_0 e^{\frac{-E}{RT_2}}C_{A2}^2, \\
  \frac{d T_2}{d t} = & \frac{F_{20}}{V_{L2}} T_{20} + \frac{F_{10}}{V_{L2}} T_1  - \frac{F_{10}+F_{20}}{V_{L_2}} T_{2} - \frac{ \Delta H}{\rho_L C_p}k_0 e^{\frac{-E}{RT_2}}C_{A2}^2 + \frac{Q_2}{\rho_L C_p V_{L_2}},
\end{align}
\end{subequations}
where $C_{Aj}$, $V_{L_j}$, and $T_j$ are the concentrations in reactor $j$, volumes of the reacting liquid, temperature in the first and the second reactor, respectively. 
Consequently, the states of the chemical process is $\xv \eqdef (C_{A1} -C_{A1s}, T_1 - T_{1s}, C_{A2} - C_{A2s}, T_2 - T_{2s})^{\sf T}$, where $C_{Ajs}$ and $T_{js}$ are the stationary concentration and temperature in reactor $j$, respectively.
The control inputs are designed as the inlet concentration of the species and the heat rate supplied by the heating jacket, that is, $\uv \eqdef ( C_{A10} - C_{A1s}, Q_1 - Q_{1s}, C_{A20} - C_{A2s}, Q_2 - Q_{2s})^{\sf T}$, where $ C_{Ajs}$ and $ Q_{js}$ are the feed concentrations and heat rates at the steady state. We use the linearized dynamics of this two-CSTR-in-series process with corresponding parameters listed in the following table, and adopt the LQR and Kalman filter for its optimal control and state observation, respectively. 
\begin{table}[h!]
\centering
\begin{tabular}{  |p{5cm}|p{5cm} |  } 
\hline
\multicolumn{2}{|c|}{Table 1 - Parameters of the two-CSTR-in-series.} \\
\hline
$T_{10} = \SI{300}{\K} $ & $T_{20} = \SI{300}{\K}$\\
$F_{10} = \SI{5}{m^3/h}$  & $F_{20} = \SI{5}{m^3 /h}$\\
$V_{L_1} = \SI{1}{m^3}$ & $V_{L_2} = \SI{1}{m^3}$\\
$T_{1s} = \SI{401.9}{K}$ & $T_{2s} = \SI{401.9}{K}$\\
$C_{A1s} = \SI{1.954}{kmol/m^3}$ & $C_{A2s} = \SI{1.954}{kmol/m^3}$\\
$C_{A10s} = \SI{4e3}{mol/m^3} $& $C_{A20s} = \SI{4e3}{mol/m^3}$\\
$Q_{1s} = \SI{0.0}{kJ/h}$& $Q_{2s} = \SI{0.0}{kJ/h}$\\
$k_0 = \SI{8.46e6}{m^3/kmol/h}$ & $\Delta H = \SI{-1.15e4}{kJ/kmol}$ \\
$C_p = \SI{0.231}{kJ/kg/K}$& $R = \SI{8.314}{kJ/kmol/K}$ \\
$\rho_L = \SI{1e3}{kg/m^3}$&$ E = \SI{5e4}{kJ/kmol }$ \\
\hline
\end{tabular}
\label{table} 
\end{table}
\subsection{Performance of Single Measurement Attacks}
Fig.~\ref{fig:single_attack_metric} depicts the performance of single measurement attacks at varying weighting parameter $\lambda$. 
{To evaluate the attack performance under both complete and incomplete knowledge of second-order statistics, we present the results in Fig.~\ref{fig:single_attack_metric} where the first column of subfigures corresponds to the results with full knowledge, while the second column corresponds to the results with incomplete knowledge. 
In attack construction with incomplete knowledge, we adopt the same system setting as in the full-knowledge scenario, except for adding a mismatch $\Delta\Sigmam_{\xv\xv}$ to the covariance matrix of state variables $\Sigmam_{\xv \xv} = \epsilonv \epsilonv^\top$ (where the components of $\epsilonv$ has a uniform distribution). 
We examine the performance in terms of the distraction of stationary distributions of states and state estimates ($D(P_{\xiv_{\av}} \|P_{\xiv})$) and that of measurements $D(P_{\yv_a} \|P_{\yv})$. We also include the KL divergence between the distributions of state estimates with and without attacks, i.e., $D(P_{\hat{\xv}_a} \| P_{\hat{\xv}})$.} 
{In general, a larger $\lambda$ leads to a smaller KL divergence $D(P_{\yv_a} \| P_{\yv})$, indicating a lower probability of attack detection, as well as a smaller $D(P_{\xiv_{\av}} \| P_{\xiv})$ that corresponds to a less deviation from the stationary distribution.}

\begin{figure}[ht]
	\centering
	\begin{subfigure}{0.45\textwidth}
		\centering
		\includegraphics[width=\textwidth]{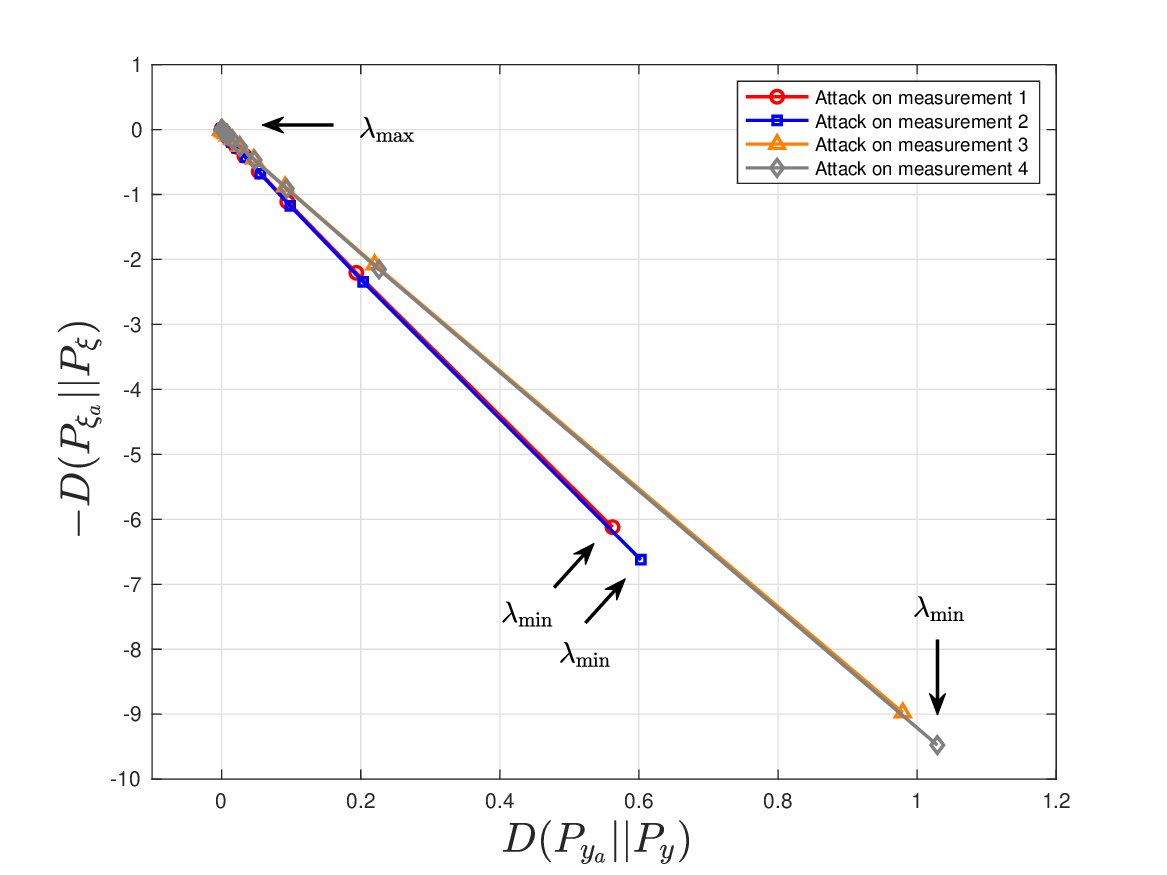} 
		\caption{Pareto curve of single measurement attacks, $-D(P_{\xiv_a} \| P_{\xiv} )$ vs $D(P_{\yv_a} \| P_{\yv} )$  with full knowledge.}
		\label{fig:single_paretocurve}
	\end{subfigure}
	\begin{subfigure}{0.45\textwidth}
		\centering
		\includegraphics[width=\textwidth]{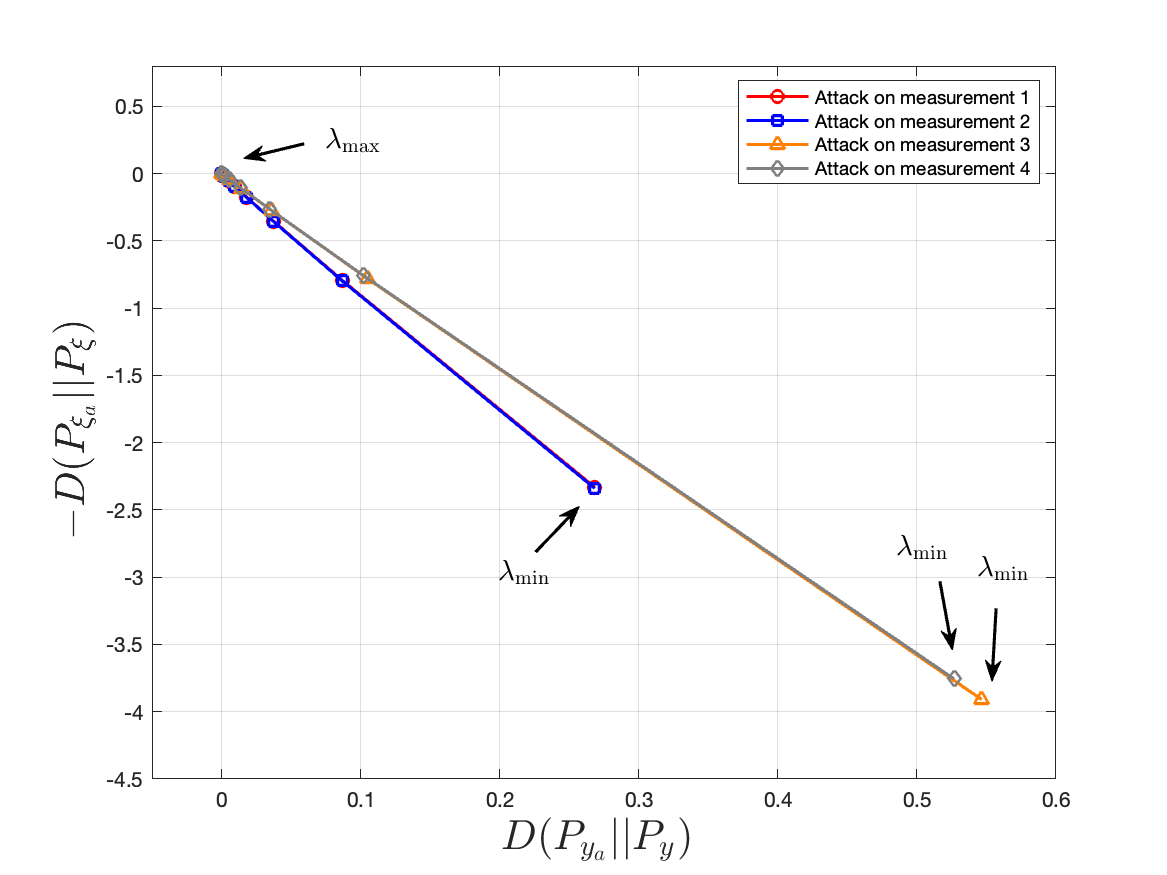} 
		\caption{Pareto curve of single measurement attacks, $-D(P_{\xiv_a} \| P_{\xiv} )$ vs $D(P_{\yv_a} \| P_{\yv} )$ with incomplete knowledge.}\label{fig:single_paretocurve2}
	\end{subfigure}
	\hfill
	\begin{subfigure}{0.45\textwidth}
		\centering
		\includegraphics[width=\textwidth]{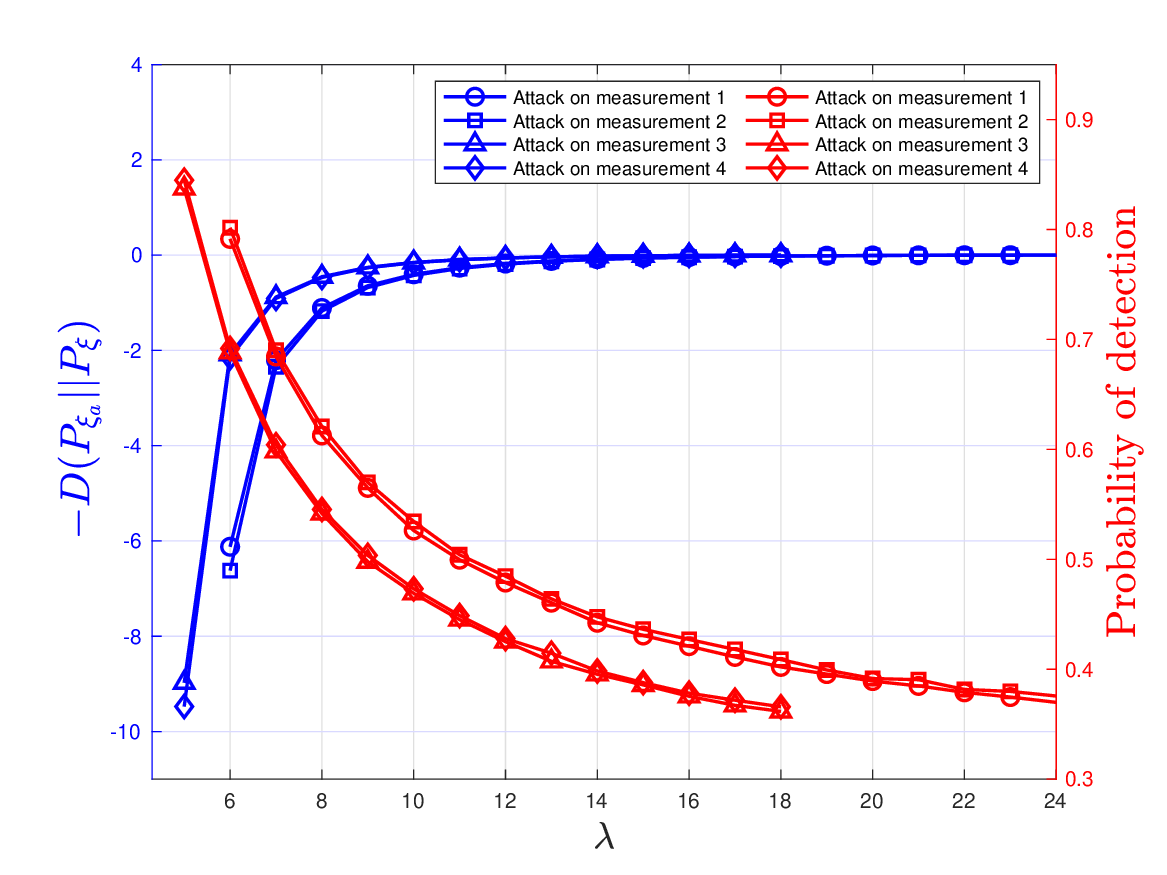} 
		\caption{Attacks in terms of $D(P_{\xiv_a} \| P_{\xiv} )$ versus probability of detection with full knowledge.}
		\label{fig:single_D1_Prob}
	\end{subfigure}
	\begin{subfigure}{0.45\textwidth}
		\centering
		\includegraphics[width=\textwidth]{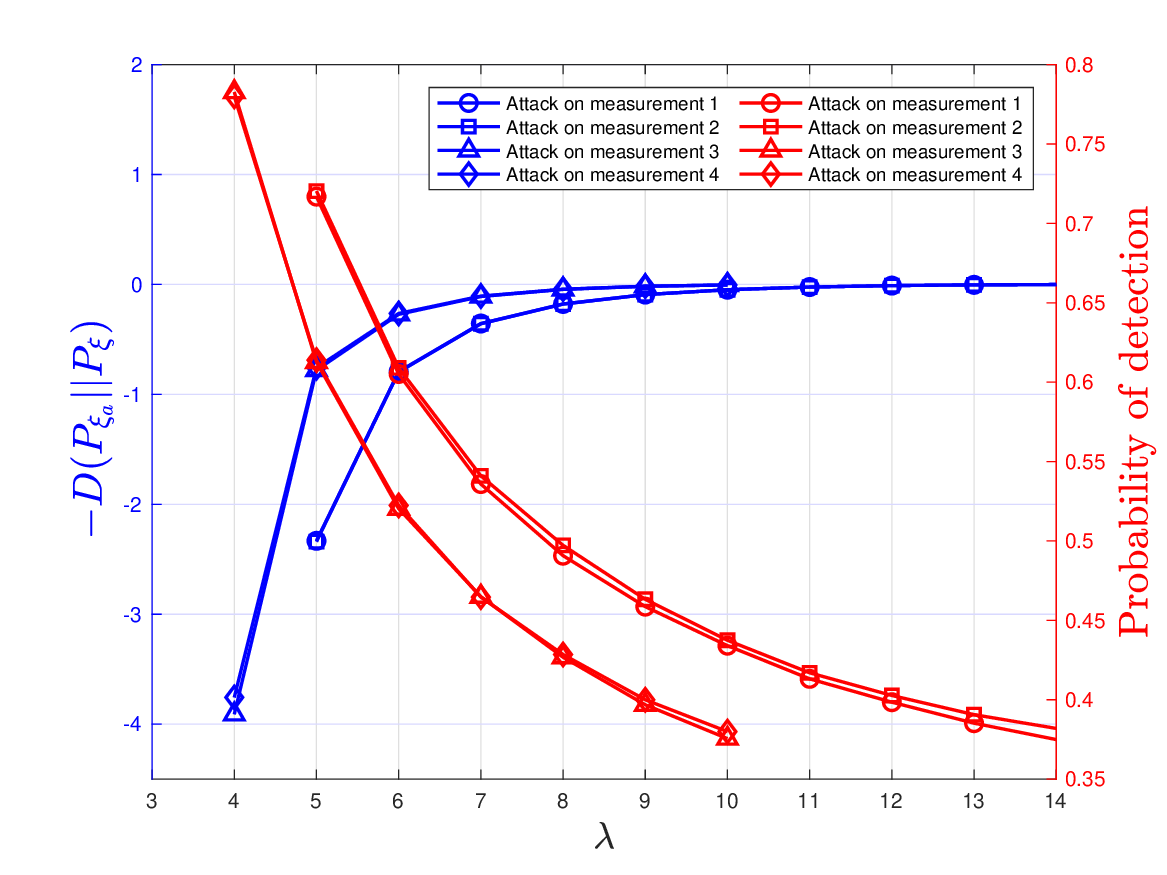} 
		\caption{Attacks in terms of $D(P_{\xiv_a} \| P_{\xiv} )$ versus probability of detection with incomplete knowledge.}\label{fig:incomplete_D1andP}
	\end{subfigure}
	\begin{subfigure}{0.45\textwidth}
		\centering
		\includegraphics[width=\textwidth]{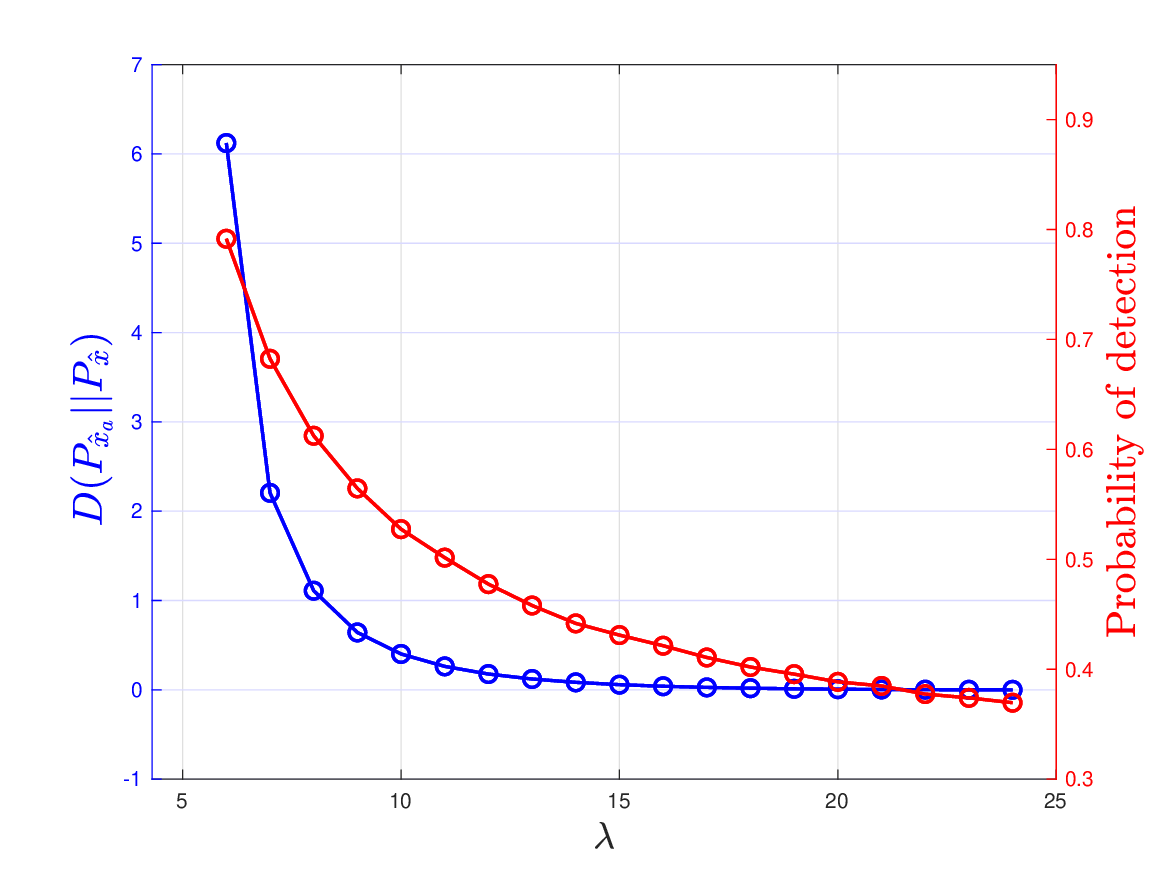} 
		\caption{Pareto curve of single measurement attacks, $ D(P_{\hat{\xv}_a} \| P_{\xv} )$ v.s. probability of detection with full knowledge.}
		\label{fig:single_lambda}
	\end{subfigure}
	\centering
	\begin{subfigure}{0.45\textwidth}
		\includegraphics[width=\textwidth]{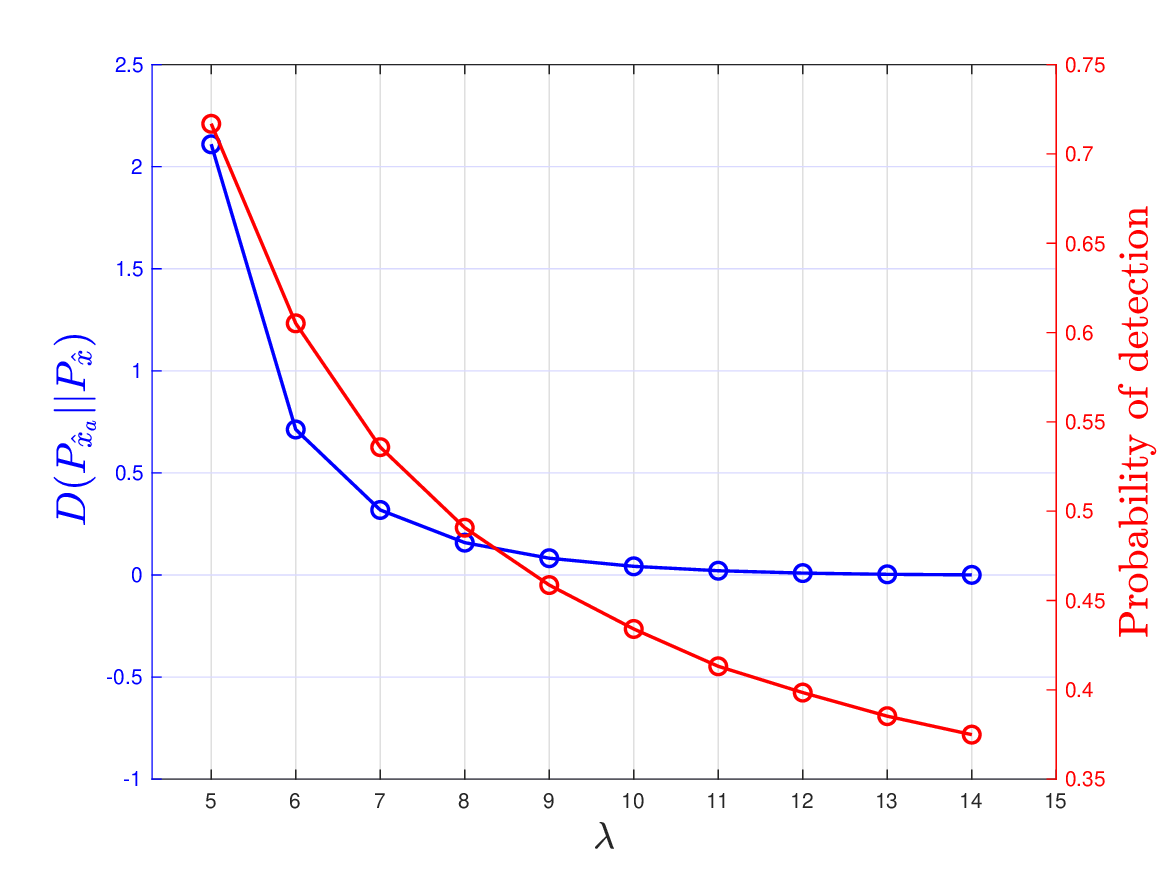} 
		\caption{Pareto curve of single measurement attacks, $ D(P_{\hat{\xv}_a} \| P_{\xv} )$ v.s. probability of detection with incomplete knowledge.}\label{fig:incomplete_DamageandP}
	\end{subfigure}
	\caption{Attack performance of single measurement attacks in terms of $\lambda$ with full knowledge and incomplete knowledge of second-order statistics. }
	\label{fig:single_attack_metric}
\end{figure}

{Fig.~\ref{fig:single_paretocurve} depicts the performance of single measurement attacks with full knowledge in terms of the tradeoff between KL divergence of $(\xv, \hat{\xv})$ distributions and KL divergence of the measurement distributions. As expected, for all measurements, both $D(P_{\yv_a} \|P_{\yv})$ and $D(P_{\xiv_{\av}} \|P_{\xiv})$ decrease monotonically when the weighting parameter $\lambda$ increases. 
The comparison between full and incomplete knowledge reveals that with perfect knowledge, the attacker can achieve a wider range of $D(P_{\yv_a} \| P_{\yv})$ and $D(P_{\xiv_a} \| P_{\xiv})$ values. Moreover, with the same probability of detection, the disruption caused by the attack with full knowledge is significantly greater.
Overall, the Pareto curves in Fig. \ref{fig:single_paretocurve} span a wider range of both $D(P_{\yv_a} \| P_{\yv})$ and $D(P_{\xiv_a} \| P_{\xiv})$, offering more flexibility to prioritize either attack disruption or stealthiness. 
In contrast, the attack constructions under incomplete knowledge of second-order statistics, shown in Fig.~\ref{fig:single_paretocurve2}, are confined to a much narrower region in both metrics. 
Furthermore, when attacking the same measurement and achieving the same $D(P_{\yv_a} \| P_{\yv})$, the attacks with incomplete knowledge result in smaller $D(P_{\xiv_a} \| P_{\xiv})$, indicating a reduced ability to alter the steady-state distribution. The same trend is evident in the comparison between Fig.~\ref{fig:single_D1_Prob} and Fig.~\ref{fig:incomplete_D1andP}.}

In Fig.~\ref{fig:single_D1_Prob} and Fig.~\ref{fig:incomplete_D1andP}, we depict the KL divergence between the distributions of the states and state estimates under attacks and without attacks, and the probability of detection in terms of different choice of $\lambda$. The probability of detection results in this section are obtained by averaging $2\times 10^4$ realizations of the measurements.
For different measurements, $j = 1,2,3,4$, the conditions on the weighting parameters $\lambda$ are different as in~\eqref{singleattack_condition_lambda}, where the inequality depends on the parameters $a_j, b_j, c_j$, and $d_j$, i.e., on the measurement itself. It is observed that increasing the value of $\lambda$ to $18$ yields a small probability of detection. Actually, the probability of detection monotonically decreases until $\lambda$ approaches its supremum. With increasing $\lambda$, the term $-D(P_{\xiv_{\av}} \|P_{\xiv})$ increases monotonically, which in turn indicates that the KL divergence between the distributions of $(\xv, \hat{\xv})$ under attacks and without attacks decreases. This is expected since when the attacker is more conservative in probability of detection, i.e., decides to give smaller thresholds to the probability of detection, the deviations that the attack can introduce reduces. 
{As expected, in both the full knowledge case in Fig.~\ref{fig:single_D1_Prob} and the incomplete knowledge case in Fig.~\ref{fig:incomplete_D1andP}, both the probability of detection and the disruption captured by $D(P_{\xiv_a} \| P_{\xiv})$ decrease as $\lambda$ increases. However, under incomplete knowledge, for the same probability of detection, the disruption is significantly smaller compared to that in Fig.\ref{fig:single_D1_Prob}. This is consistent with the intuition that better knowledge enables the attacker to alter the distribution of steady state under same probability of detection.}

{A similar pattern is observed in Fig.~\ref{fig:single_lambda} and Fig.~\ref{fig:incomplete_DamageandP}, where attack performance is compared in terms of $D(P_{\hat{\xv}_a} \| P_{\hat{\xv}})$ and the probability of detection. Specifically, for the same probability of detection, the disruption captured by $D(P_{\hat{\xv}_a} \| P{\hat{\xv}})$ is greater in the full knowledge case as shown in Fig.~\ref{fig:single_lambda} than in the incomplete knowledge case in Fig.~\ref{fig:incomplete_DamageandP}.} In Fig.~\ref{fig:single_lambda}, we specifically present the KL divergence between the distributions of the state estimates $\hat{\xv}$ under attacks and without attacks. From a practical point of view, the attacker can be interested to see the estimates before and after attacks, i.e., $\hat{\xv}_a$ and $\hat{\xv}$, as this is a direct metric of the effect of the attacks. {When $\lambda$ increases, both $D(P_{\hat{\xv}_a}\|P_{\hat{\xv}})$ and the probability of detection decrease. This indicates that from the attack detection perspective, the attacker benefits from increasing $\lambda$ at the cost of reducing the deviation of the stationary distribution. We also observed that both the decrease shrink with larger $\lambda$, which means that the effect of changing $\lambda$ becomes smaller when $\lambda$ is large.
As we can see, there is a significant decrease on $D(P_{\hat{\xv}_a}\|P_{\hat{\xv}})$ and probability of detection in when $\lambda$ increases from $6$ to $7$. However, when $\lambda$ is larger than $11$, the decrease of $D(P_{\hat{\xv}_a}\|P_{\hat{\xv}})$ is almost neglectable, while the attacker still obtains a considerably stable decrease on the resulting probability of detection.}
 
To illustrate the stationary distributions of the chemical process, we depict the distributions of the concentration and temperature in tank $1$ and tank $2$, respectively, when attacking the first measurement with $\lambda = 8$. Fig.~\ref{fig:single_steady_dist_tank1_attack} and Fig.~\ref{fig:single_steady_dist_tank1_noattack} depict the difference of these stationary distributions with attacks and without attacks in tank $1$, and Fig.~\ref{fig:single_steady_dist_tank2_attack} and Fig.~\ref{fig:single_steady_dist_tank2_noattack} depict for tank $2$. As shown in the figures, both the distribution of concentration and temperature significantly differ.

\begin{figure}[ht]
	\centering
	\begin{subfigure}{0.45\textwidth}
		\centering
		\includegraphics[width=\textwidth]{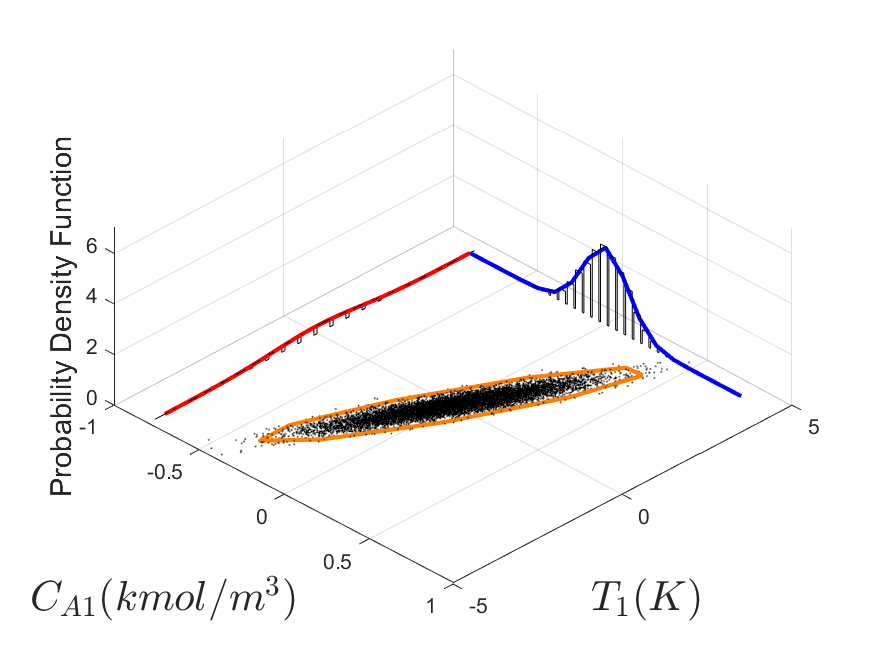} 
		\caption{Stationary distributions of $C_{A1}$ and $T_{A1}$ in tank $1$ with single measurement attacks when $\lambda = 8$. }
		\label{fig:single_steady_dist_tank1_attack}
	\end{subfigure}
	\hfill
	\begin{subfigure}{0.45\textwidth}
	\centering
	\includegraphics[width=\textwidth]{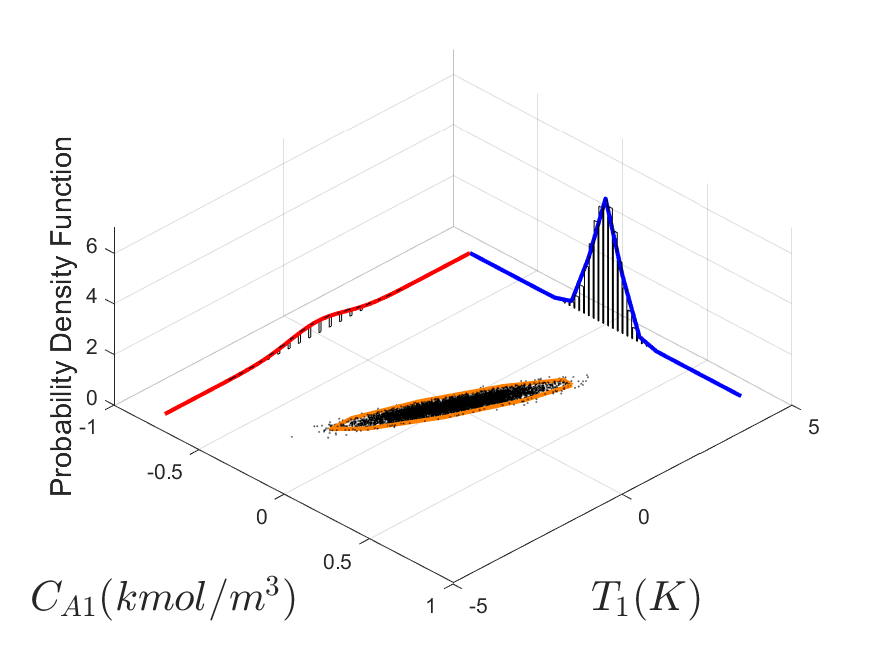} 
	\caption{Stationary distributions of $C_{A1}$ and $T_{A1}$ in tank $1$ without attacks. }
	\label{fig:single_steady_dist_tank1_noattack}
    \end{subfigure}
    	\hfill
    \begin{subfigure}{0.45\textwidth}
    	\centering
    	\includegraphics[width=\textwidth]{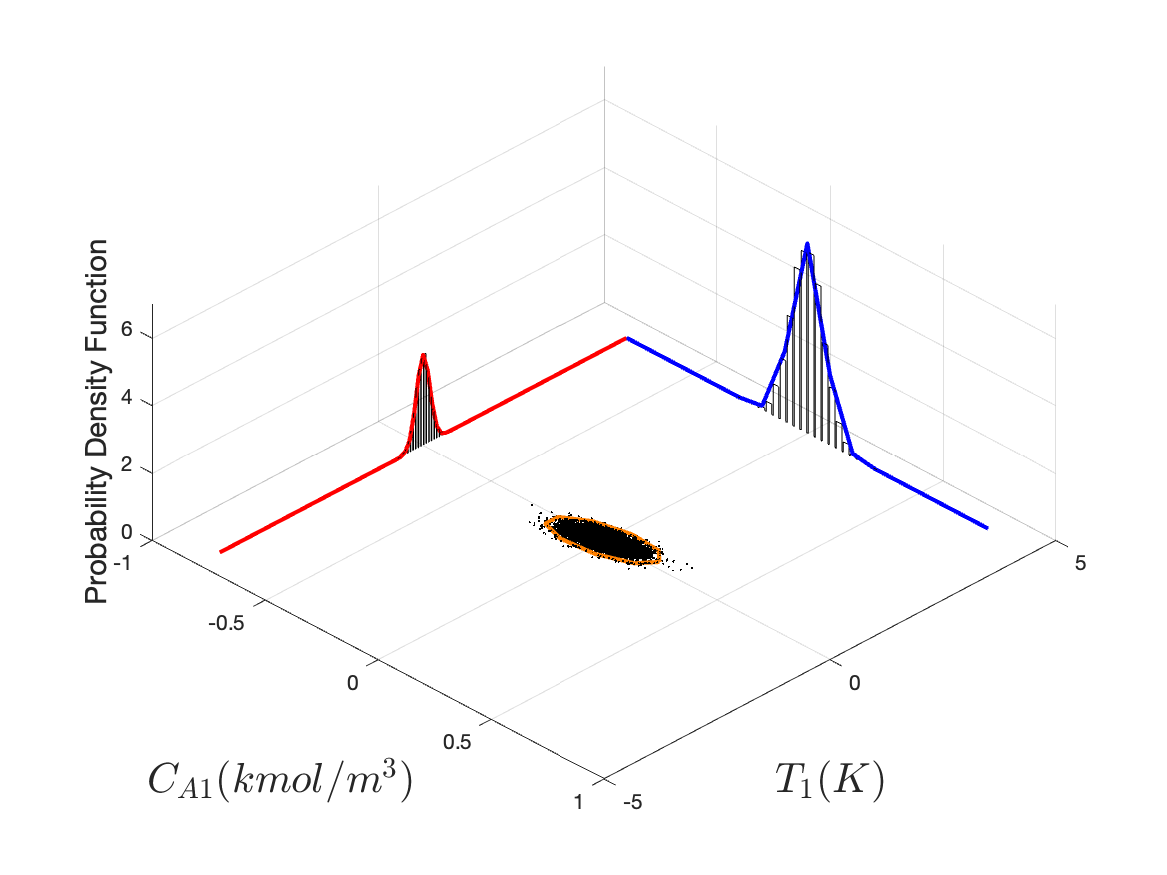} 
    	\caption{Stationary distributions of $C_{A2}$ and $T_{A2}$ in tank $2$ with single measurement attacks when $\lambda = 8$.}
    	\label{fig:single_steady_dist_tank2_attack}
    \end{subfigure}
        \hfill
    	\begin{subfigure}{0.45\textwidth}
    	\centering
    	\includegraphics[width=\textwidth]{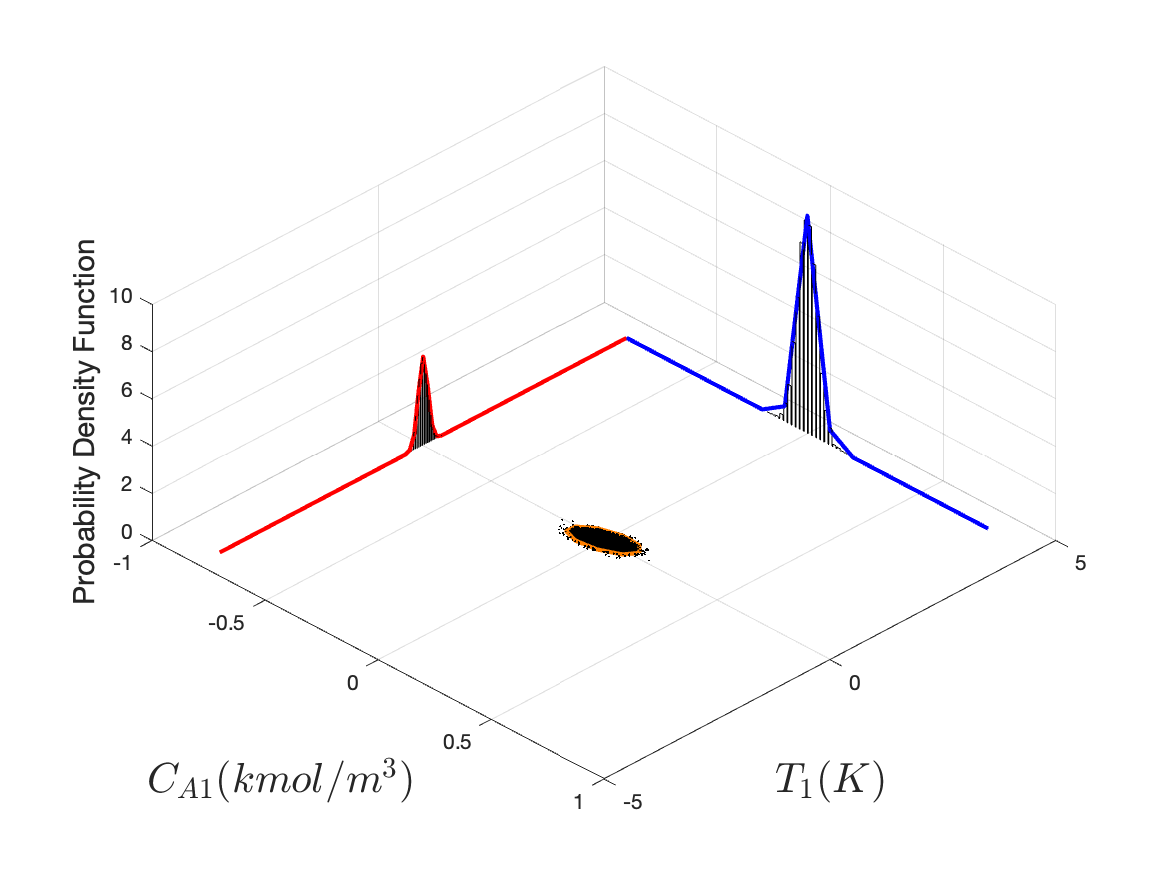} 
    	\caption{Stationary distributions of $C_{A2}$ and $T_{A2}$ in tank $2$ without attacks. }
    	\label{fig:single_steady_dist_tank2_noattack}
    \end{subfigure}
	\caption{Analysis of single measurement attacks in terms of stationary distributions.}
	\label{fig:single_attack_steady_dist}
\end{figure}

\subsection{Performance of $k$-sparse Attacks}
{In the sequential process of constructing $k$-sparse attacks, the procedure of attacking one measurement at each step is consistent with the analysis of single measurement attacks in the previous section. We still obtain the same effect of $\lambda$ on the deviations of stationary distributions and the probability of detection.}
In this section, we evaluate the performance of $k$-sparse attacks with a focus on assessing the effect of the sparsity constraint $k$. Intuitively, a large $k$ implies a more powerful attacker.

Fig.~\ref{fig:k_performance} depicts the performance of the $k$-sparse attacks for different $k$. 
Fig.~\ref{k_paretocurve} shows the performance of the $k$-sparse attack construction in Algorithm~\ref{alg:greedy} in terms of the tradeoff between the {KL divergence of states' and state estimates' distributions and KL divergence of the measurements}, i.e., $D(P_{\xiv_{\av}} \|P_{\xiv})$ and $D(P_{\yv_a}\|P_{\yv})$. Note that in $k$-sparse attack construction, when $k>1$, the Pareto curve is obtained by considering the feasible $\lambda$ for all measurement $j = 1,2,3,4$.
In general, larger $k$ yields a lower Pareto curve, as the ability to attack more measurements yields better performance in attack constructions overall. 
As expected, larger values of the parameter $\lambda$ yield smaller values of KL divergence $D(P_{\xiv_{\av}} \|P_{\xiv})$, i.e., the probability of detection is prioritized in the construction over the disruption for $k = 1, 2, 3, 4$. Moreover, with the same $D(P_{\xiv_{\av}} \|P_{\xiv})$, smaller values of $k$ yield smaller disruption captured by $D(P_{\xiv_{\av}} \|P_{\xiv})$. This coincides with the intuition that with the same probability of detection, the ability to compromise more measurements facilitates the disruption on the stationary distribution. 
 \begin{figure}[ht]
	\centering
	\begin{subfigure}{0.45\textwidth}
		\centering
		\includegraphics[width=\textwidth]{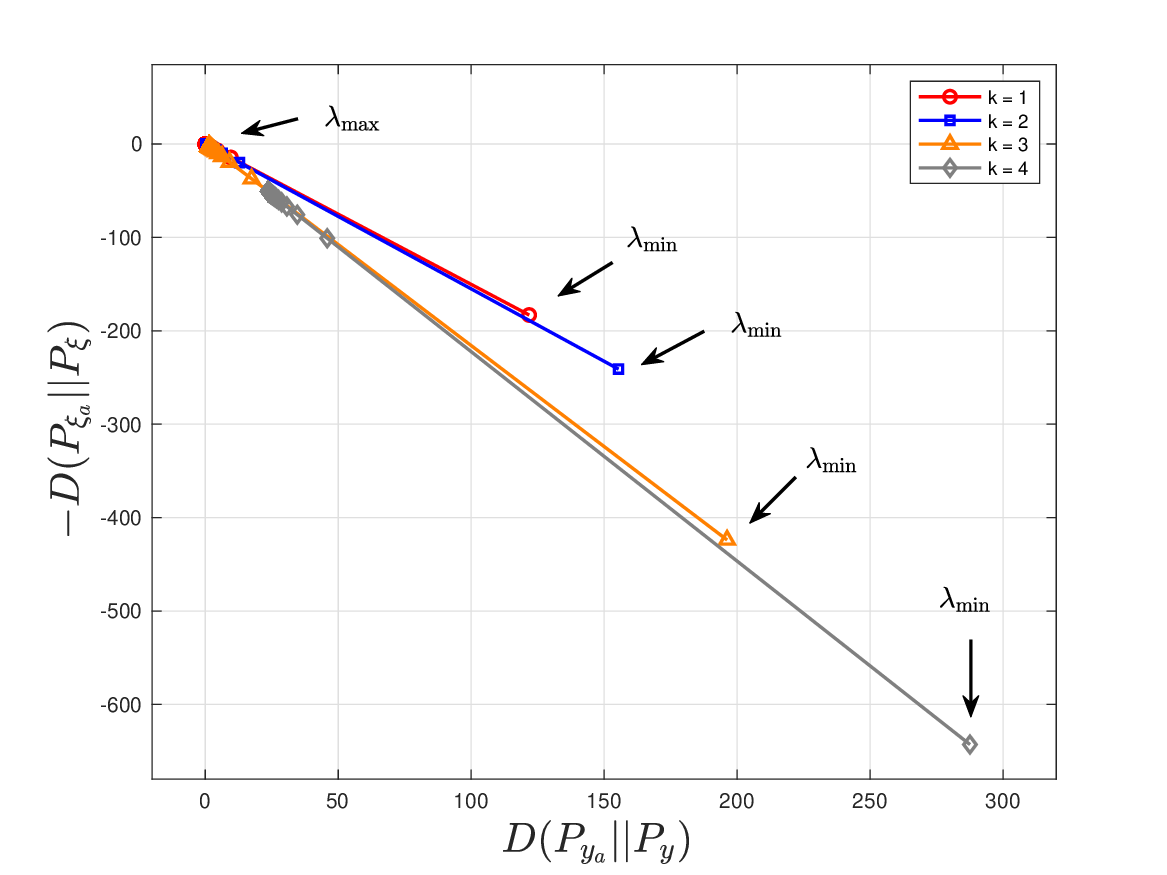} 
		\caption{Performance of the $k$-sparsity attacks in terms of $D(P_{ \xiv_a} \| P_{ \xiv} )$ and $D(P_{ \yv_a} \| P_{ \yv} )$ for different $k$ in the greedy algorithm.}
		\label{k_paretocurve}
	\end{subfigure}
	\begin{subfigure}{0.45\textwidth}
		\centering
		\includegraphics[width=\textwidth]{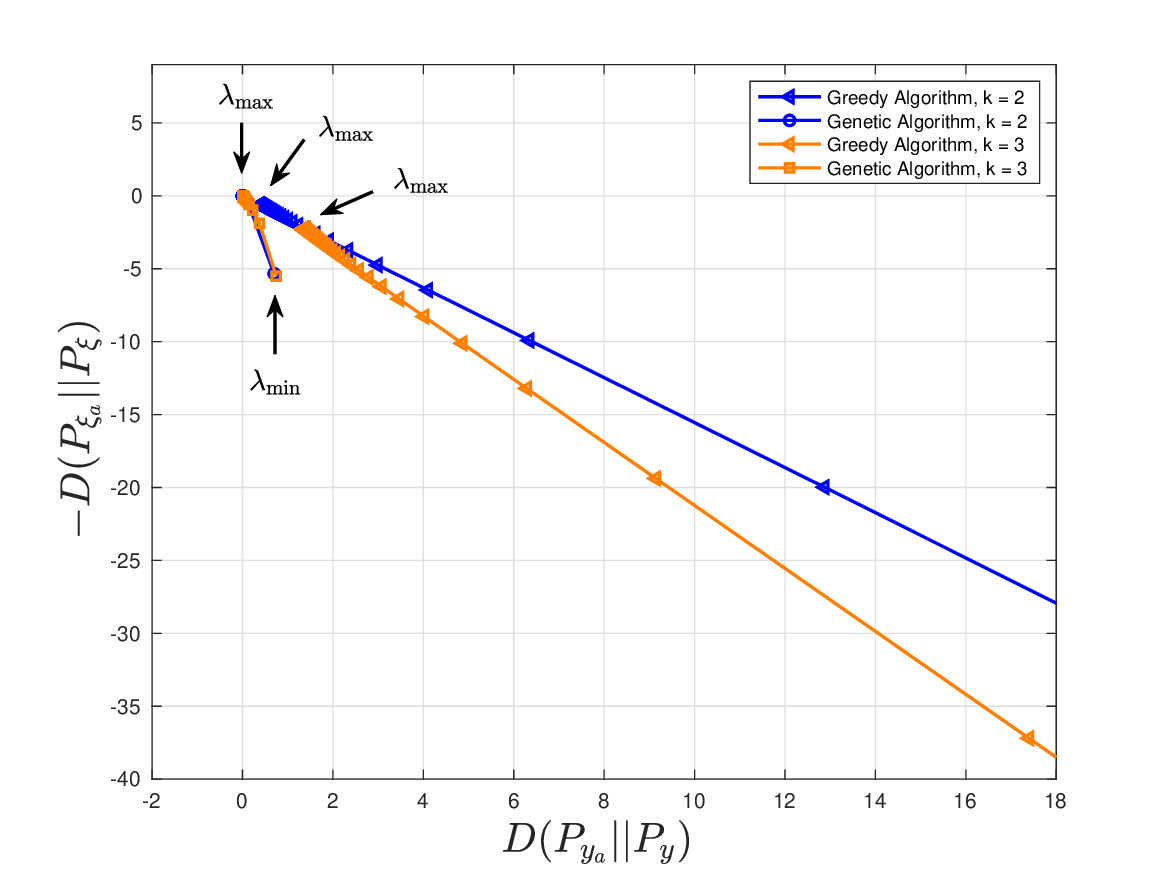} 
		\caption{Performance of the $k$-sparsity attacks in terms of $D(P_{ \xiv_a} \| P_{ \xiv} )$ and $D(P_{ \yv_a} \| P_{ \yv} )$ for different $k$ in the greedy algorithm and genetic algorithm when $k = 2$ and $k = 3$.}
		\label{k_paretocurve_two_algo}
	\end{subfigure}
	\begin{subfigure}{0.45\textwidth}
		\centering
		\includegraphics[width=\textwidth]{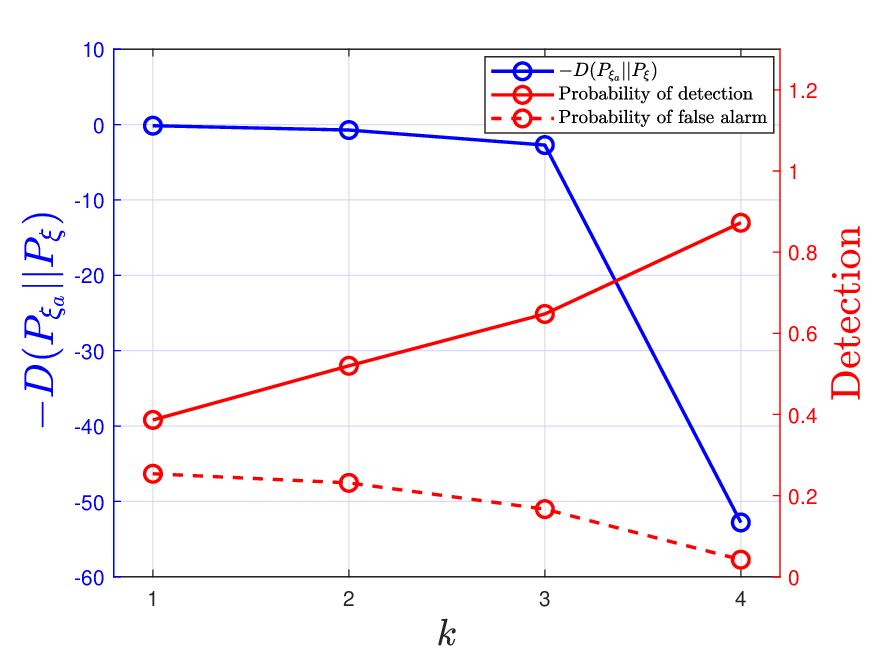} 
		\caption{Performance of the $k$-sparsity attacks in terms of $D(P_{ \xiv_a} \| P_{ \xiv} )$ and probability of detection for different $k$ in the greedy algorithm.}
		\label{fig:k_D_and_Pd}
	\end{subfigure}
	\caption{Analysis of $k$-sparse attacks.}
	\label{fig:k_performance}
\end{figure}

{To assess the performance of the proposed attack constructions, we adopt a genetic algorithm (GA) and present the results in Fig.~\ref{k_paretocurve_two_algo} for $k = 2, 3$. In the genetic algorithm, we initialize with all possible individuals and apply the optimal attack variances to the measurements as in eq.~\eqref{v_star}, thereby avoiding crossover and mutation in the domain of the variances that are continuously valued. 
Overall, the genetic algorithm yields small values of $D(P_{\yv_a} \| P_{\yv})$, which is favorable, as they indicate low probabilities of detection. However, the attack disruption, captured by $D(P_{\xiv_a} \|P_{\xiv})$, is also small. This is because the attack construction in the genetic algorithm is constrained to a much smaller feasible range, preventing the attacker from designing attacks that focus more on disruption. 
In addition to this limitation, the genetic algorithm generally suffers from a significant increase in the number of possible individuals as $m$ increases, along with the need for tuning multiple parameters. These challenges make it difficult for the algorithm to converge efficiently, rendering it impractical for real-world applications. Instead, it serves as a benchmark to demonstrate that the proposed greedy algorithm performs comparatively well, with lower computational complexity and a much larger feasible range of $\lambda$. From the attacker's perspective, the proposed approach is more practical, as it is implemented sequentially. }

In Fig.~\ref{fig:k_D_and_Pd}, the performance of the $k$-sparse attack in Algorithm~\ref{alg:greedy} is evaluated in terms of the disruption $D(P_{\xiv_{\av}} \|P_{\xiv})$ and the probability of detection for different $k$. 
As expected, larger values of the parameter
$\lambda$ yield smaller probability of detection while increasing the KL divergence between the stationary distributions under attacks and without attacks. We note that the probability of attack detection increases approximately linearly with respect
to $k$. Simultaneously in this range of $k$, $-D(P_{\xiv_{\av}} \|P_{\xiv})$ decreases approximately linear with respect to $k$. It is worth noting that the probability of false alarm maintains a relatively stable value for all $k$.

Similarly as in the case of single measurement attacks, we illustrate the stationary distributions of the variables in tank $1$ and tank $2$, respectively. Fig.~\ref{fig:k_steady_dist_tank1_attack} and Fig.~\ref{fig:k_steady_dist_tank1_noattack} depict the stationary distributions with attacks in comparison to the distributions without attacks in tank $1$. Fig.~\ref{fig:k_steady_dist_tank2_attack} and Fig.~\ref{fig:k_steady_dist_tank2_noattack} depict the stationary distributions with attacks in comparison to the distributions without attacks in tank $2$. As shown in the figures, in both tanks, all the distributions of variables with attacks differ significantly from the their distributions without attacks.

\begin{figure}[ht]
	\centering
		\begin{subfigure}{0.45\textwidth}
		\centering
    \includegraphics[width=\textwidth]{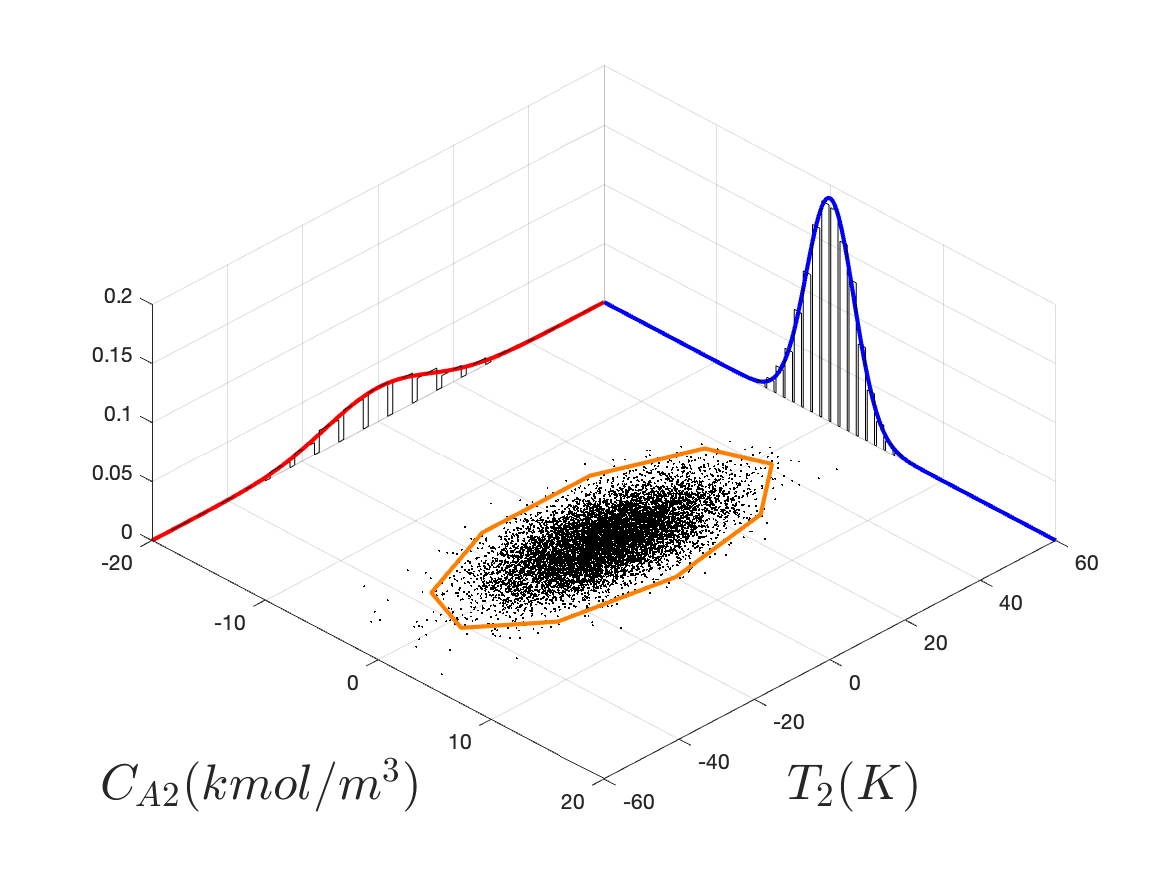} 
		\caption{Stationary distributions of $C_{A1}$ and $T_{A1}$ in tank $1$ of $k$-sparse attacks with $k=4$ and $\lambda = 8$.}
		\label{fig:k_steady_dist_tank1_attack}
	\end{subfigure}
	\begin{subfigure}{0.45\textwidth}
		\centering
		\includegraphics[width=\textwidth]{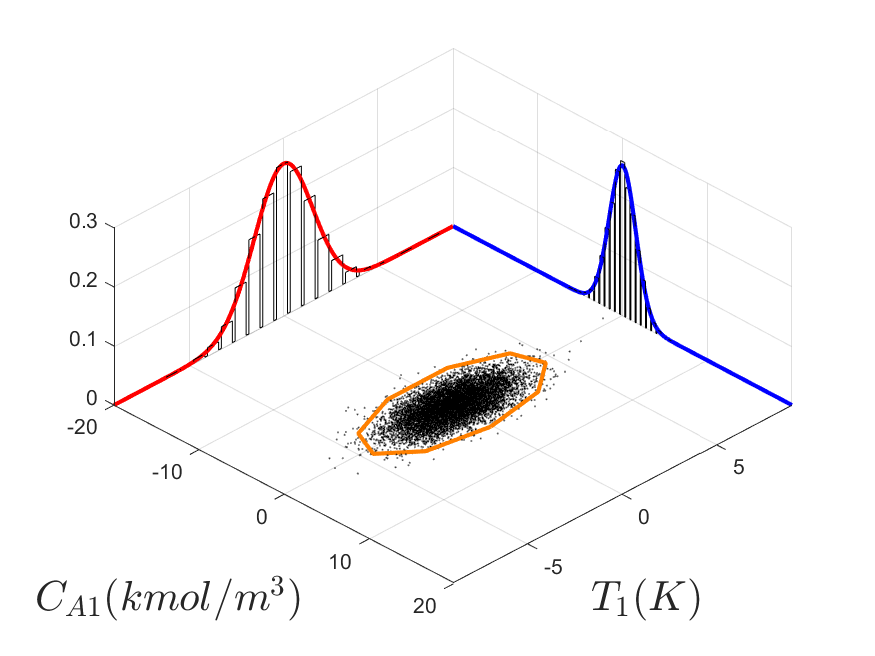} 
		\caption{Stationary distributions of $C_{A1}$ and $T_{A1}$ in tank $1$ without attacks. }
		\label{fig:k_steady_dist_tank1_noattack}
	\end{subfigure}
	\centering
	\begin{subfigure}{0.45\textwidth}
		\centering
		\includegraphics[width=\textwidth]{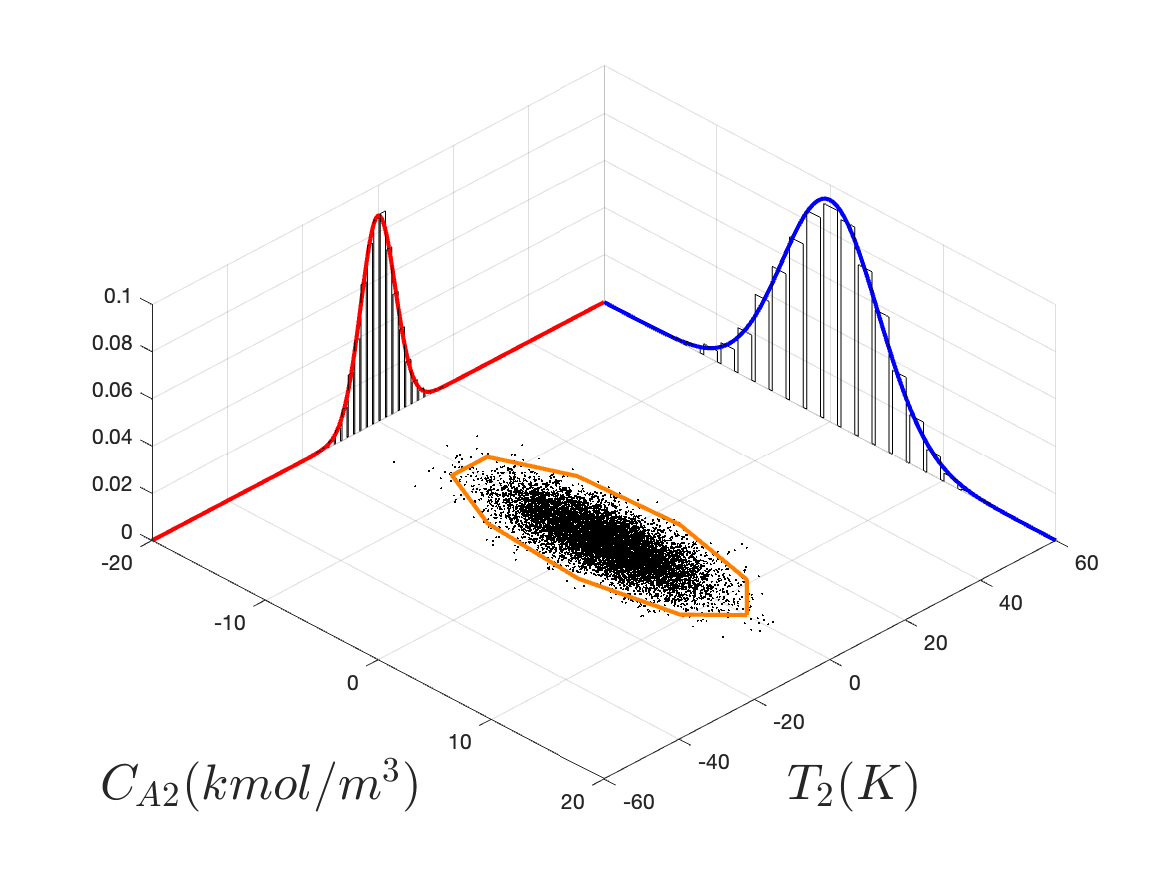} 
		\caption{Stationary distributions of $C_{A2}$ and $T_{A2}$ in tank $2$ of $k$-sparse attacks with $k=4$ and $\lambda = 8$.}
		\label{fig:k_steady_dist_tank2_attack}
	\end{subfigure}
	\begin{subfigure}{0.45\textwidth}
		\centering
		\includegraphics[width=\textwidth]{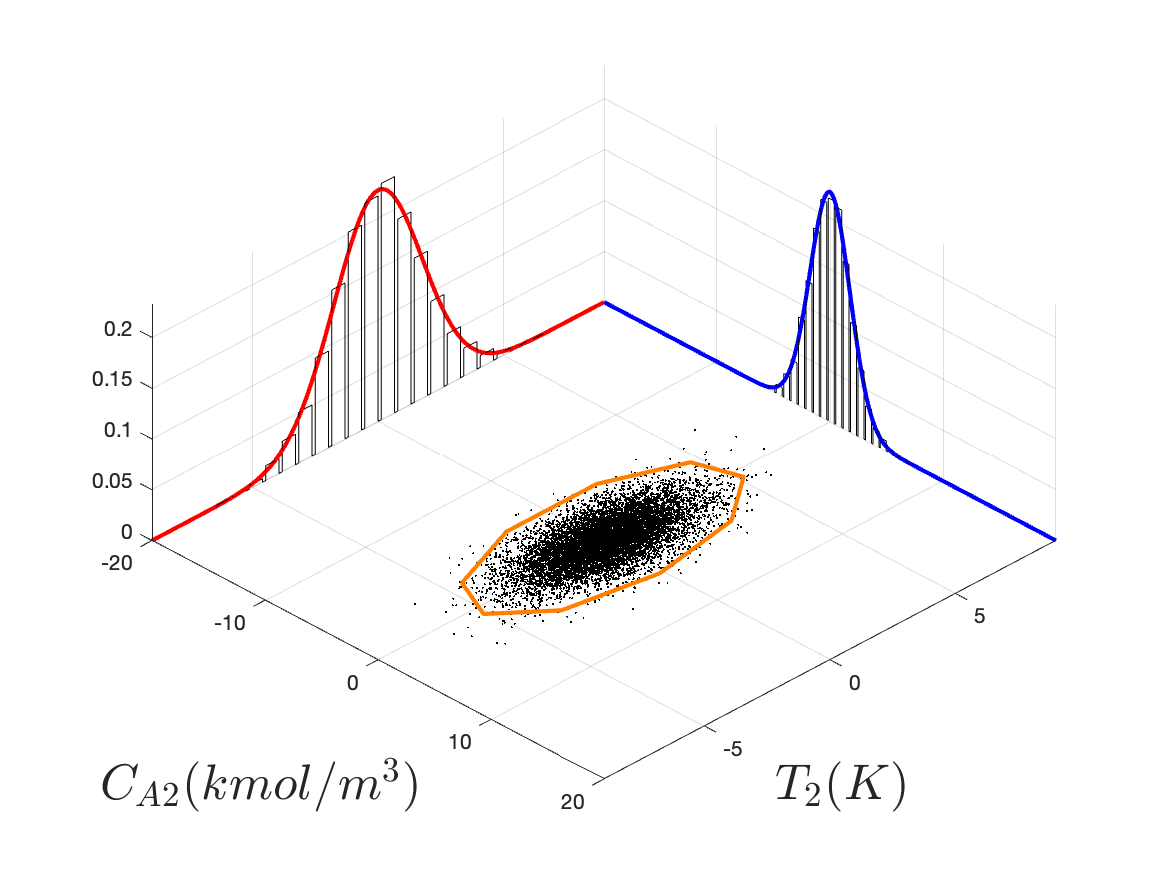} 
		\caption{Steady state distributions of $C_{A2}$ and $T_{A2}$ in tank $2$ without attacks. }
		\label{fig:k_steady_dist_tank2_noattack}
	\end{subfigure}
		\caption{Analysis of $k$-sparse attack in terms of stationary distributions.}
	\label{fig:k_attack_steady_dist}
\end{figure}

\subsection{Discussion on Vulnerability and Resilience}
In single measurement attacks, it is interesting to observe that in Fig.~\ref{fig:single_paretocurve}, attacking different measurements resulting in Pareto curves that do not intersect with each other. This observation leads to the insight that {\emph{some measurements are more vulnerable to DIAs}}. Particularly, in Fig.~\ref{fig:single_paretocurve}, the Pareto curves of attacking measurement $1$ and $2$, namely $C_{A1}$ and $T_1$, are lower than the ones corresponding to attacking measurements $3$ and $4$, namely $C_{A2}$ and $T_2$. Actually, the Pareto curve for attacking measurement $1$ is nearly identical to the one for attacking measurement $2$. Likewise, the Pareto curve for attacking measurement $3$ is nearly identical to the one for attacking measurement $4$. This is intuitively expected as the measurements within the same tank suffer from DIAs to a extent. 

More importantly, we observed the measurements within tank $1$ obtain lower Pareto curves in comparison with the ones within tank $2$.
Therefore, we conclude in this chemical process, tank $1$ is more vulnerable to DIAs. In Fig.~\ref{fig:single_D1_Prob}, the KL divergence $D(P_{ \xiv_a} \| P_{ \xiv} )$ for attacking measurement $1$ is nearly identical to the one for attacking measurement $2$. A similar pattern is observed for measurement $3$ and $4$. The same behavior exists in the probability of detection. Attacking measurement $1$ and $2$ yields higher probability of detection as attacking measurement $3$ and $4$. This coincides with the conclusion from Fig.~\ref{fig:single_paretocurve} that measurement $1$ and $2$ are more vulnerable. The curve of $D(P_{ \xiv_a} \| P_{ \xiv} )$ in Fig.~\ref{fig:single_D1_Prob} also suggests that attacking measurement $1$ and $2$ yields larger deviation of the distributions of the states and their estimates. Hence, we conclude with tank $1$ is more vulnerable to DIAs from both metrics. 

{Hence, we conclude that tank $1$ is more vulnerable to DIAs. This analysis also highlights a key aspect of process system resilience. That is, the resilience to maintain safe and stable operation in the presence of DIAs. In the context of chemical processes, resilience not only involves recovering from faults or attacks but also minimizing the deviation of critical variables such as concentrations, flow rates, or temperatures from their nominal operating ranges. The fact that attacks on tank $1$ lead to larger deviations in stationary distributions and higher probabilities of detection suggests that as part of the process, tank $1$ is more sensitive to DIAs and has a lower resilience margin compared to tank $2$, making it a critical subsystem that could compromise the overall safety and performance of the chemical process under attack. In practical terms, this insight inspires enhanced monitoring, stronger anomaly detection, or the inclusion of redundant sensing and control strategies for tank $1$. By quantifying this localized vulnerability, the analysis provides a pathway toward targeted resilience enhancement in chemical process systems.}

\section{Conclusions}\label{sec_conclusion}
\par {This paper proposes a systematic information-theoretic framework to assess and quantify the process resilience of control systems under optimal DIAs. The resilience is assessed by the attack disruption and detection by the proposed optimal DIAs, where the attack construction is cast as an optimization problem.} An analytical solution is obtained for single measurement attacks. For the proposed $k$-sparse attack where the attacker is able to attack $k$ measurements, a heuristic greedy algorithm is proposed, where one-at-a-time measurement is attacked in a sequential process. The greedy step results in a convex optimization problem under appropriate conditions, which can be solved efficiently and yields a low-complexity attack
updating rule. We numerically evaluate the attack
performance in a two-reactor chemical process, and
we concluded that some measurements in the process are more vulnerable to DIAs than others in the sense of a worse disruption-stealth tradeoff.
 
In this study, we adopted a conservative approach by assuming the attacker has complete knowledge of the system model and the second-order moments of the variables in the process. Additionally, the attack construction was restricted to a linear system. {While this work incorporates a sparsity constraint as well as numerical studies where the attacker is with limited knowledge, further exploration remains open to address scenarios analytically and to extend the methodology to nonlinear systems.}


\section*{Acknowledgments}
\noindent The authors express their gratitude to the Department of Chemical and Biomolecular Engineering at NC State University for providing the startup funding and to Prof. Aritra Mitra in the Department of Electrical and Computer Engineering at NC State University for insightful discussions and for directing us to valuable literature on control-theoretic approaches to cyberattacks.
\section*{Author Contributions}
\noindent \textbf{Xiuzhen Ye}: Formal analysis; investigation; methodology; software; writing -- original draft, writing -- editing. \\
\noindent \textbf{Wentao Tang}: Conceptualization; formal analysis; funding acquisition; investigation; methodology; project administration; writing -- review. 
\section*{Data Availability Statement}
\noindent All the data to create the figures are available in the supplementary materials. These figures are created from computer simulations. 
All the Matlab codes needed to reproduce the simulations are available at GitHub repository: \url{https://github.com/WentaoTang-Pack/DIAonProcessSystems/}. All the figures are provided in the GitHub repository.

\section*{ORCID}
\noindent \textit{Xiuzhen Ye}\quad \url{https://orcid.org/0000-0002-7859-713X} \\
\noindent \textit{Wentao Tang}\quad \url{https://orcid.org/0000-0003-0816-2322} 

\let\OLDthebibliography\thebibliography
\renewcommand\thebibliography[1]{
  \OLDthebibliography{#1}
  \setlength{\itemsep}{0pt plus 0.3ex}
}
\bibliography{bib}

\begin{thebibliography}{10}
\providecommand{\url}[1]{\texttt{#1}}
\providecommand{\urlprefix}{URL }

\bibitem{Alur_PrinciplesCPS_2015}
Alur R.
\newblock \emph{Principles of Cyber-Physical Systems}.
\newblock The MIT Press. 2015.

\bibitem{HR_IoT_17}
Hunzinger R.
\newblock {SCADA} fundamentals and applications in the {I}o{T}.
\newblock \emph{Internet of things and data analytics handbook}.
  2017;\hspace{0pt}pp. 283--293.

\bibitem{RH_ICT_11}
Baheti R, Gill H.
\newblock Cyber-physical systems.
\newblock \emph{The impact of control technology}.
  2011;\hspace{0pt}12(1):161--166.

\bibitem{YG_IoT_2023}
Yu Z, Gao H, Cong X, Wu N, Song HH.
\newblock A survey on cyber-physical systems security.
\newblock \emph{IEEE Internet of Things Journal}.
  2023;\hspace{0pt}10(24):21670--21686.

\bibitem{YM_CCC_2009}
Mo Y, Sinopoli B.
\newblock Secure control against replay attacks.
\newblock In: \emph{2009 47th Annual Allerton Conference on Communication,
  Control, and Computing (Allerton)}. IEEE. 2009;\hspace{0pt} pp. 911--918.

\bibitem{DZ_JAS_2022}
Duo W, Zhou M, Abusorrah A.
\newblock A survey of cyber attacks on cyber physical systems: Recent advances
  and challenges.
\newblock \emph{IEEE/CAA J Autom Sin}. 2022;\hspace{0pt}9(5):784--800.

\bibitem{Helen_Mathematics_2018}
Durand H.
\newblock A nonlinear systems framework for cyberattack prevention for chemical
  process control systems.
\newblock \emph{Mathematics}. 2018;\hspace{0pt}6(9):169.

\bibitem{KeHelen_CERD_2021}
Rangan KK, Oyama H, Durand H.
\newblock Integrated cyberattack detection and handling for nonlinear systems
  with evolving process dynamics under Lyapunov-based economic model predictive
  control.
\newblock \emph{Chem Eng Res Des}. 2021;\hspace{0pt}170:147--179.

\bibitem{WuHelen_Mathematics_2018}
Wu Z, Albalawi F, Zhang J, Zhang Z, Durand H, Christofides PD.
\newblock Detecting and handling cyber-attacks in model predictive control of
  chemical processes.
\newblock \emph{Mathematics}. 2018;\hspace{0pt}6(10):173.

\bibitem{OyamaHelen_AICHE_2020}
Oyama H, Durand H.
\newblock Integrated cyberattack detection and resilient control strategies
  using Lyapunov-based economic model predictive control.
\newblock \emph{AIChE J}. 2020;\hspace{0pt}66(12).

\bibitem{HelenMatthew_Mathematics_2020}
Durand H, Wegener M.
\newblock Mitigating safety concerns and profit/production losses for chemical
  process control systems under cyberattacks via design/control methods.
\newblock \emph{Mathematics}. 2020;\hspace{0pt}8(4):499.

\bibitem{Huang_TIT_2018}
Huang K, Zhou C, Tian YC, Yang S, Qin Y.
\newblock Assessing the physical impact of cyberattacks on industrial
  cyber-physical systems.
\newblock \emph{IEEE Trans Ind Electr}. 2018;\hspace{0pt}65(10):8153--8162.

\bibitem{LY_TISSEC_11}
Liu Y, Ning P, Reiter MK.
\newblock False data injection attacks against state estimation in electric
  power grids.
\newblock \emph{ACM Trans Inform Syst Secur}. 2011;\hspace{0pt}14(1):1--33.

\bibitem{DH_TSMC_2020}
Ding D, Han QL, Ge X, Wang J.
\newblock Secure state estimation and control of cyber-physical systems: A
  survey.
\newblock \emph{IEEE Trans Systems, Man, and Cybernetics: Systems}.
  2020;\hspace{0pt}51(1):176--190.

\bibitem{WS_TC_2018}
Wu G, Sun J, Chen J.
\newblock Optimal data injection attacks in cyber-physical systems.
\newblock \emph{IEEE Trans Cybernetics}. 2018;\hspace{0pt}48(12):3302--3312.

\bibitem{WT_ACC_2022}
Tang W, Daoutidis P.
\newblock Data-driven control: Overview and perspectives.
\newblock In: \emph{American Control Conference}. IEEE. 2022;\hspace{0pt} pp.
  1048--1064.

\bibitem{PW_CCE_23}
Parker S, Wu Z, Christofides PD.
\newblock Cybersecurity in process control, operations, and supply chain.
\newblock \emph{Computers \& Chemical Engineering}.
  2023;\hspace{0pt}171:108--169.

\bibitem{YM_SCS_2010}
Mo Y, Sinopoli B.
\newblock False data injection attacks in control systems.
\newblock In: \emph{Proc. of the Workshop on Secure Control Systems}.
  Stockholm, Sweden. 2010;\hspace{0pt} .

\bibitem{FabioP_CDC_2011}
Pasqualetti F, D{\"o}rfler F, Bullo F.
\newblock Cyber-physical attacks in power networks: Models, fundamental
  limitations and monitor design.
\newblock In: \emph{50th IEEE Conference on Decision and Control and European
  Control Conference}. IEEE. 2011;\hspace{0pt} pp. 2195--2201.

\bibitem{FabioP_TAC_2013}
Pasqualetti F, D{\"o}rfler F, Bullo F.
\newblock Attack detection and identification in cyber-physical systems.
\newblock \emph{IEEE Trans Autom Control}. 2013;\hspace{0pt}58(11):2715--2729.

\bibitem{FabioP_TSG_2018}
Amini S, Pasqualetti F, Mohsenian-Rad H.
\newblock Dynamic load altering attacks against power system stability: Attack
  models and protection schemes.
\newblock \emph{IEEE Trans Smart Grid}. 2016;\hspace{0pt}9(4):2862--2872.

\bibitem{ShilpaEllis_ACC_2022}
Narasimhan S, El-Farra NH, Ellis MJ.
\newblock Controller switching-enabled active detection of multiplicative
  cyberattacks on process control systems.
\newblock In: \emph{American Control Conference}. IEEE. 2022;\hspace{0pt} pp.
  2473--2478.

\bibitem{ShilpaEllis_JPC_2022}
Narasimhan S, El-Farra NH, Ellis MJ.
\newblock Active multiplicative cyberattack detection utilizing controller
  switching for process systems.
\newblock \emph{{J Process Control}}. 2022;\hspace{0pt}116:64--79.

\bibitem{Huang_CSR_2020}
Huang X, Kroening D, Ruan W, Sharp J, Sun Y, Thamo E, Wu M, Yi X.
\newblock A survey of safety and trustworthiness of deep neural networks:
  Verification, testing, adversarial attack and defence, and interpretability.
\newblock \emph{Comput Sci Rev}. 2020;\hspace{0pt}37:100270.

\bibitem{WZ_IECR_25}
Wu G, Zhang H, Wu W, Wang Y, Wu Z.
\newblock Physics-Informed Graph Convolutional Recurrent Network for
  Cyber-Attack Detection in Chemical Process Networks.
\newblock \emph{Industrial \& Engineering Chemistry Research}.
  2025;\hspace{0pt}64(6):3370--3382.

\bibitem{WW_CERD_24}
Wu G, Wang Y, Wu Z.
\newblock Physics-informed machine learning in cyber-attack detection and
  resilient control of chemical processes.
\newblock \emph{Chemical Engineering Research and Design}.
  2024;\hspace{0pt}204:544--555.

\bibitem{FabioP_ACC_2015}
Bai CZ, Pasqualetti F, Gupta V.
\newblock Security in stochastic control systems: Fundamental limitations and
  performance bounds.
\newblock In: \emph{American Control Conference}. Chicago, US.
  2015;\hspace{0pt} pp. 195--200.

\bibitem{FabioP_Automatica_2017}
Bai CZ, Pasqualetti F, Gupta V.
\newblock Data-injection attacks in stochastic control systems: Detectability
  and performance tradeoffs.
\newblock \emph{Automatica}. 2017;\hspace{0pt}82:251--260.

\bibitem{TM_ElementsofIT}
Cover TM, Thomas JA.
\newblock \emph{Elements of information theory}.
\newblock John Wiley \& Sons. 1999.

\bibitem{YE_TSG_21}
Ye X, Esnaola I, Perlaza SM, Harrison RF.
\newblock Stealth Data Injection Attacks with Sparsity Constraints.
\newblock \emph{IEEE Trans Smart Grid}. 2023;\hspace{0pt}14(4):3201--3209.

\bibitem{YE_SGC_20}
Ye X, Esnaola I, Perlaza SM, Harrison RF.
\newblock Information theoretic data injection attacks with sparsity
  constraints.
\newblock In: \emph{IEEE International Conference on Communications, Control,
  and Computing Technologies for Smart Grids}. Phoenix, USA. 2020;\hspace{0pt}
  pp. 1--6.

\bibitem{JN_LRT_33}
Neyman J, Pearson ES.
\newblock On the Problem of the Most Efficient Tests of Statistical Hypotheses.
\newblock \emph{Phil Trans R Soc London}. 1933;\hspace{0pt}231:289--337.

\bibitem{OM_TNNLS_16}
Ozay M, Esnaola I, Yarman~Vural FT, Kulkarni SR, Poor HV.
\newblock Machine Learning Methods for Attack Detection in the Smart Grid.
\newblock \emph{IEEE Trans Neur Netw Learn Syst}.
  2016;\hspace{0pt}27(8):1773--1786.

\bibitem{YE_IETSG_22}
Ye X, Esnaola I, Perlaza SM, Harrison RF.
\newblock An information theoretic vulnerability metric for data integrity
  attacks on smart grids.
\newblock \emph{IET Smart Grid}. 2024;\hspace{0pt}7(5):583--592.

\bibitem{seber}
Seber GA.
\newblock \emph{A matrix handbook for statisticians}.
\newblock John Wiley \& Sons. 2008.

\bibitem{CW_CERD_2021}
Chen S, Wu Z, Christofides PD.
\newblock Cyber-security of centralized, decentralized, and distributed
  control-detector architectures for nonlinear processes.
\newblock \emph{Chem Eng Res Des}. 2021;\hspace{0pt}165:25--39.

\end{thebibliography}
\end{document}